\documentclass[runningheads,orivec]{llncs}
\usepackage{amsmath}
\usepackage[T1]{fontenc}
\usepackage{graphicx}
\usepackage{bussproofs}
\usepackage{proof}
\usepackage{doc}
\usepackage{amssymb}    
\usepackage{mathtools}
\usepackage{cite}
\usepackage{xcolor,color}
\usepackage[all]{xy}
\usepackage{esvect}
\usepackage{stmaryrd}
\usepackage{units}
\usepackage{stackrel}
\usepackage{yfonts,upgreek}
\usepackage{relsize,bm}

\usepackage{hyperref}

\newtheorem{exm}[theorem]{Example}

\newtheorem{cor}[theorem]{Corollary}
\newtheorem{defi}[theorem]{Definition}

\DeclareSymbolFont{AMSb}{U}{msb}{m}{n}
\DeclareMathSymbol{\N}{\mathbin}{AMSb}{"4E}
\DeclareMathSymbol{\Z}{\mathbin}{AMSb}{"5A}
\DeclareMathSymbol{\R}{\mathbin}{AMSb}{"52}
\DeclareMathSymbol{\Q}{\mathbin}{AMSb}{"51}
\DeclareMathSymbol{\I}{\mathbin}{AMSb}{"49}
\DeclareMathSymbol{\C}{\mathbin}{AMSb}{"43}

\newcommand{\lk}{{\bf{LK}}}
\newcommand{\LK}{{\bf LK}}
\newcommand{\LKp}{{\bf LK}^{+}}
\newcommand{\LKpp}{{\bf {LK}^{++}}}

\newcommand{\lkqs}{{\bf{LK}_{\texttt{shift}}}}

\newcommand{\lj}{{\mathrm{LJ}}}

\newcommand{\LJpp}{{\bf LJ}^{++}}

\newcommand{\LJp}{{\bf LJ}^{+}}

\newcommand\limp{\to}

\def\LJ{\mathbf{LJ}}

\def\phi{\varphi}

\def\SW{\mathrm{SW}}
\def\CD{\mathrm{CD}}
\def\ED{\mathrm{ED}}

\newcommand{\iqc}{{\sf IQC}}

\newcommand{\SA}{{{S}_A}} %Andrews skolemization
\newcommand{\LE}{{\bf L}\varepsilon}
\newcommand{\IE}{{\bf I}\varepsilon}

\title{Proof Theory of Quantifier Shifts -  An Overview}
\subtitle{(Dedicated to V.P. Orevkov)}
\author{Matthias Baaz\inst{1}\orcidID{0000-0002-7815-2501}\thanks{This research was funded in part by the Austrian Science Fund (FWF) 10.55776/P36571.} \and
Mariami Gamsakhurdia\inst{1}\orcidID{0009-0009-0468-9945} }
\authorrunning{M. Baaz, M. Gamsakhurdia}
\institute{Institute of Discrete Mathematics and Geometry, TU Wien, Austria \\\email{\{baaz,mariami\}@logic.at}
}
\date{July 2026}

\begin{document}

\maketitle

\begin{abstract}
 Quantifier shifts are one of the oldest proof methods, but they are orthogonal to the cut free format:
They can shorten proofs non-elementarily. (The proof is based on the sequence of proofs exhibiting a non-elementary difference between proofs with cuts and cut-free proofs by Orevkov). The cost is, however, that those proofs are globally sound but locally unsound. (Their subderivations might be unsound in contrast to the Hilbert concept of proof). In non-classical logics these calculi are only sound when all quantifier shifts are added. This characterizes all intermediate logics where Skolemization and the epsilon calculus is conservative (in contrast to the epsilon theorem which holds only for finite-valued Gödel logics).

% keywords can be removed
\end{abstract}

\section{Introduction}

%Hilbert's epsilon calculus \cite{hilbertbernays39} is one of the oldest proof-theoretic formalism with the first epsilon Theorem (in modern sense the extended first epsilon Theorem) as the first proper proof-theoretic result (the epsilon calculus has been identified as the only reasonable formalization for mathematics by Bourbaki \cite{Bourbaki}).

Hilbert's epsilon calculus \cite{hilbertbernays39} is one of the earliest proof-theoretic formalisms. Its first epsilon theorem (more precisely, the extended first epsilon theorem in modern terminology) constitutes one of the first genuine proof-theoretic results. The epsilon calculus was even regarded by Bourbaki \cite{Bourbaki} as the most natural formalization of mathematics.

The epsilon formalism is based on the replacement of $\exists xA(x)$ by $A(\varepsilon_x A(x))$ and $\forall xA(x)$ by $A(\varepsilon_x \neg A(x))$. Any Hilbert-style proof in a system based on the axiom  $A(t) \rightarrow \exists xA(x)$ and the inference rules
 \begin{center}
\AxiomC{$A(a)    \rightarrow B $}
\UnaryInfC{$\exists xA(x)    \rightarrow B$}
\DisplayProof
\end{center}
can be directly translated to epsilon calculus.
The same holds if $\forall$ is in the language and for the sequent-format. The only requirement is a weak regularity condition: this means that whenever the same characteristic variables are used, the same quantified formula in the conclusion follows.
In due course of such a translation, the formulas $A(t) \rightarrow \exists xA(x)$ are replaced by $A(t) \rightarrow A(\varepsilon_x A(x))$ and characteristic variables are simply substituted by suitable epsilon terms throughout the proof.
Consequently, the translation is independent of whether the underlying propositional calculus is finitely or infinitely axiomatized. Moreover, unlike ordinary first-order calculi, the full deduction theorem remains valid. Consequently, every derivation in the epsilon calculus can be identified with a propositional implication whose antecedent is the conjunction of all critical formulas and whose consequent is the epsilon translation of the final formula.
%This implies that it does not matter whether there are finitely or infinitely many propositional axioms and that the full deduction theorem holds, contrary to the usual first-order formalism. As a consequence, every derivation in the epsilon calculus may be identified with a valid propositional formula, namely the conjunction of critical formulas implying the epsilon translation of the result. \\

    In classical logic $\tau_xA(x)$ could be used for duality to represent  $\forall x A(x)$ instead of $A(\varepsilon_x \neg A(x))$. In this case, we have to add critical formulas $A(\tau_xA(x))\rightarrow A(t)$ directly representing $A(\varepsilon_x \neg A(x))\rightarrow A(t)$.  

The critical formulas immediately yield the classical quantifier-shift principles. For example:
\begin{enumerate}
    \item $\forall x (A(x)\vee B)\rightarrow \forall xA(x)\vee B$
generated by $A(\tau_x (A(x)\vee B))    \rightarrow A(\tau_xA(x))\vee B$.
    \item $\forall xA(x)\vee B\rightarrow\forall x (A(x)\vee B) $ generated by $A(\tau_x A(x))\vee B    \rightarrow  A(\tau_x (A(x)\vee B))$ and 
\begin{center}
\AxiomC{$C    \rightarrow D$}
\UnaryInfC{$C\vee B\rightarrow D\vee B$}
\DisplayProof
\end{center}
        \item $\forall x A(x)    \rightarrow B\rightarrow (\exists xA(x)    \rightarrow B)$ generated by $(A(\tau_x (A(x))    \rightarrow B)    \rightarrow A(\varepsilon_x (A(x)    \rightarrow B))    \rightarrow B$.
    \item $\exists x (A(x)    \rightarrow B)\rightarrow (\forall xA(x)    \rightarrow B)$ generated by $A(\tau_x A(x))\rightarrow A(\varepsilon_x (A(x)    \rightarrow B))    \rightarrow  B)$ and
    \begin{center}
\AxiomC{$C    \rightarrow D$}
\UnaryInfC{$D    \rightarrow B\rightarrow C    \rightarrow B$}
\DisplayProof
\end{center}
    \item $\forall x A(x)    \rightarrow B\rightarrow \exists x(A(x)    \rightarrow B)$ generated by $A(\tau_x (A(x))    \rightarrow B))    \rightarrow A(\varepsilon_x (A(x)    \rightarrow B)    \rightarrow B)$.

    \item $(A\rightarrow \exists x B(x))\rightarrow \exists x(A    \rightarrow B(x))$ generated by $(A    \rightarrow B(\varepsilon_x (B(x))    \rightarrow (A    \rightarrow B(\varepsilon_x (A    \rightarrow B(x)))$.
    \item $\exists x (A    \rightarrow B(x))\rightarrow A    \rightarrow \exists B(x)$ generated by $A    \rightarrow B(\varepsilon_x (A    \rightarrow B(x))    \rightarrow A    \rightarrow B(\varepsilon_x (B(x))$ and 
\begin{center}
\AxiomC{$B    \rightarrow C$}
\UnaryInfC{$(A    \rightarrow B)\rightarrow (A    \rightarrow C)$}
\DisplayProof
\end{center}
    \item   $\wedge, \vee$ similar, non similar when $\neg A$ is considered as $A    \rightarrow \bot$.
\end{enumerate}

 Although Hilbert's epsilon calculus is inherently classical, one may ask what consequences the epsilon formalism has in non-classical logics, i.e., the addition of critical formulas to the propositional axioms of these non-classical logics. From the arguments above it can only be conservative for logics which admit all classically valid quantifier shifts (see the example above). Note, however, that contrary to an argument in Bell\cite{Bell2006},
the propositional structure is never changed as we might simply delete all term-position in the proof and the critical formulas become propositional formulas of the form $A \limp A$.\\

    The first epsilon theorem induces in addition the existence of propositional valid formulas of the form $ A_1 \rightarrow A_2 \vee A_2 \rightarrow A_3 \dots A_{n-1} \rightarrow A_n $, consequently, in intermediate logics the first epsilon theorem is only valid for finite-valued G\"odel logics in Baaz, Zach \cite{MatthiasZach}.\\

%By translating proofs into epsilon calculus via substituting characteristic variables, we observe that the epsilon terms replacing those characteristic variables might occur in the conclusion of those inference rules. Therefore, we will weaken the characteristic variable conditions and this is the core of the article.

When proofs are translated into the epsilon calculus, characteristic variables are replaced by epsilon terms. These epsilon terms may subsequently occur outside the scope of the original quantifier inference, including in the conclusion of the derivation. This observation motivates a relaxation of the traditional characteristic variables conditions and forms the starting point of the present paper.

\section{Sound Proofs With Possible Unsound Subproofs}\label{SectionQI}

Contemporary mathematics is based on the axiomatic method developed by Hilbert \cite{hilbert1899grundlagen}; A \emph{proof} in a formal system is a finite sequence of formulas
\[
A_1, A_2, \dots, A_n
\]
such that each $A_i$ is either (i) an axiom of the system, or (ii) obtained from preceding formulas $A_j$ ($j<i$) by an application of an inference rule. The last formula $A_n$ is the \emph{result} or \emph{theorem} of the proof.

Although this concept is rather recent, the alternatives have been completely forgotten. An alternative is, e.g., the méthode de descente infinie\footnote{Natural deduction can be also considered as a non-hilberian proof system, however, an easy transcription transforms it into a hilbert-style proof system.} of Fermat cf. Edwards\cite{edwards1977fermat}. In the méthode de descente infinie potential counterexamples are transformed in smaller counterexamples, according to a certain order and, separately, this order is shown to be well-founded (The méthode de descente infinie is nowadays partially reflected in the concept of circular proofs %\textbf{[cite]}
).
The axiomatic method induces, however, an asymmetry: whereas soundness is demonstrated by simply induction completeness based on more involved arguments.

In this article, we aim for a more basic representation of this phenomenon using proof-theoretic calculi. We concentrate on first-order logic. The key observation is that characteristic variables in epsilon calculus substituted by epsilon terms are allowed to occur in the conclusion. We concentrate on cut-free proofs measured by size (number of symbols). Recall that a function on the natural numbers is \textit{elementary} if it can be defined by a quantifier-free formula from $+$, $\cdot$, and the function $x\mapsto 2^x$. The main results deal with non-elementary differences between different concepts of cut-free proofs. %\textbf{A non-elementary function is a function not bound by any elementary function.}

Consider the following argument:
\begin{enumerate}
\item \label{Int1} That Kurt G\"odel is Austrian entails that Kurt G\"odel is Austrian.
\item \label{Int2} Hence, that Kurt G\"odel is Austrian entails that everyone is Austrian.
\item \label{Int3} That is, if Kurt G\"odel is Austrian, then all people are Austrian.
\item \label{Int4} Therefore, there exists a person such that, if that person is Austrian, then all people are Austrian.
\end{enumerate}
This argument is globally sound but locally unsound and can be formalized as an instance of the following proof schema:
\begin{equation}  \label{ProofIntro}
\infer[]{\vdash \exists x\,\big( A(x) \to \forall y\, A(y)\big)}
{
	\infer[]{\vdash A(a) \to \forall y\, A(y)}
	{
		\infer[]{A(a) \vdash \forall y\, A(y)}{A(a) \vdash A(a)}
	}
}
\end{equation}
%\subsection{Quantifier Inferences}

We consider derivations in sequent calculi. Sequents are expressions of the form $\Gamma \vdash \Delta$, where $\Gamma$ and $\Delta$ represent collections of formulae. The sequent $\Gamma \vdash \Delta$ is usually interpreted as `if all of $\Gamma$ hold, then at least one of the formulae in $\Delta$ holds.' Equivalently, if $\Gamma = A_0, A_1, \hdots, A_n$ and $\Delta = B_0, B_1, \hdots, B_m$, then $\Gamma \vdash \Delta$ is interpreted as: 
\begin{equation}\label{seqdef}
\bigvee_{i\leq n} \lnot A_i \vee \bigvee_{i\leq m} B_i
\end{equation}

Inferences are expressions of the form 
\begin{equation*}
\infer[]{\Gamma' \vdash \Delta'}{\Gamma \vdash \Delta}
\end{equation*}
which are to be interpreted as `the sequent $\Gamma' \vdash \Delta'$ follows from the sequent $\Gamma \vdash \Delta$.' Derivations are trees whose leaves are axioms and whose non-leaf nodes are sequents obtained from their predecessors by inferences; they are to be interpreted as \textit{proofs} of the root sequent. In fact, we frequently refer to derivations as `proofs.'

If a formula is changed by an inference, we say that the formula is an \textit{auxiliary} formula of the inference. If so, then the resulting formula in the conclusion is called the \textit{critical} formula. An inference is a \textit{quantifier inference} if the (only) critical formula has a quantifier as its outermost logical symbol. Normally, quantifier inferences consist of substituting a formula\footnote{We will frequently use $Q$ to denote any unspecified quantifier.} $Qx\, A(x)$ for an instance thereof in a sequent, e.g., as in 
\begin{equation*}
\infer[]{\Gamma \vdash \Delta, \forall x\, A(x)}{\Gamma \vdash \Delta,  A(a)}
\end{equation*}
The \textit{polarity} of a quantifier in a sequent $\Gamma \vdash \Delta$ is defined as follows: rewrite $\Gamma \vdash \Delta$ in the form \eqref{seqdef} and then rewrite the resulting expression without implication symbols. A quantifier is \textit{positive} if it is universal (resp. existential) and under the scope of an even (resp. odd) number of negation symbols, and \textit{negative} otherwise.
%a quantifier is \textit{positive} in a formula rewritten without implication symbols if it is on the right-hand side (resp. left-hand side) of a sequent and under the scope of an even (resp. odd) number of negation symbols, and \textit{negative} otherwise. 
A quantifier is \textit{strong} if it is positive and on the right-hand side of a sequent or negative and on the left-hand side of a sequent; it is \textit{weak} otherwise. 

If an inference yields a strongly-quantified formula $Qx\, A(x)$ from $A(a)$, where $a$ is a free variable, we say that $a$ is the \textit{characteristic variable} of the inference. We will denote free variables by letters $a, b, c, ...$ (in contrast, we denote bound variables by $x, y, z, ...$). We denote closed terms by $t, s, r, ...$ and variants thereof.

Let $\pi$ be any derivation. We say $b$ is a \textit{side-variable} of $a$ in $\pi$ (written $a <_\pi b$) if $\pi$ contains a strong-quantifier inference of the form:
\begin{equation*}
\infer[]{\Gamma \vdash \Delta, \forall x\, A(x, b, \vec c)}{\Gamma \vdash \Delta,  A(a, b, \vec c)}
\end{equation*}
or of the form:
\begin{equation*}
\infer[]{\exists x\, A(x, b, \vec c), \Gamma \vdash \Delta}{A(a, b, \vec c), \Gamma \vdash \Delta}
\end{equation*}

\begin{defi}[Suitable quantifier inference]
We say a quantifier inference is \emph{suitable for a proof $\pi$} if either it is a weak-quantifier inference, or the following three conditions are satisfied:
\begin{itemize}
\item (substitutability) the characteristic variable does not appear in the conclusion of $\pi$.
\item (side-variable condition) the relation $<_\pi$ is acyclic.
\item (weak regularity) the characteristic variable is not the characteristic variable of another strong-quantifier inference in $\pi$.
\end{itemize}
\end{defi}

We will work with various sequent calculi extending \textit{Propositional $\mathrm{LK}$}, the sequent calculus whose only axioms are $\bot \vdash$ and $A \vdash A$ (with $A$ atomic) and whose rules are:
\begin{align*}
\text{Structural rules:}\\
	\infer[WL]{\Gamma, A \vdash \Delta}{\Gamma \vdash \Delta} && 
	\infer[WR]{\Gamma \vdash \Delta, A}{\Gamma \vdash \Delta} \\
	\infer[CL]{\Gamma, A \vdash \Delta}{\Gamma, A, A \vdash \Delta} && 
	\infer[CR]{\Gamma \vdash \Delta, A}{\Gamma \vdash \Delta, A, A} \\	
	\infer[EL]{\Gamma_1, B, A, \Gamma_2 \vdash \Delta}{\Gamma_1, A, B, \Gamma_2 \vdash \Delta} && 
	\infer[ER]{\Gamma \vdash \Delta_1, B, A, \Delta_2}{\Gamma \vdash \Delta_1, A, B, \Delta_2} \\
	\infer[Cut]{\Gamma \vdash \Delta}{\Gamma \vdash \Delta, A & A, \Gamma \vdash \Delta} &&\\
\end{align*}
\begin{align*}
\text{Logical rules:}\\
	\infer[\wedge_1L]{\Gamma, A\wedge B \vdash \Delta}{\Gamma, A \vdash \Delta} && 
	\infer[\vee_1R]{\Gamma \vdash \Delta, A\vee B}{\Gamma \vdash \Delta, A} \\
	\infer[\wedge_2L]{\Gamma, A\wedge B \vdash \Delta}{\Gamma, B \vdash \Delta} && 
	\infer[\vee_2R]{\Gamma \vdash \Delta, A\vee B}{\Gamma \vdash \Delta, B} \\
	\infer[\vee L]{\Gamma, A\vee B \vdash \Delta}{\Gamma, A \vdash \Delta & \Gamma, B \vdash \Delta} && 
	\infer[\wedge R]{\Gamma \vdash \Delta, A\wedge B}{\Gamma \vdash \Delta, A & \Gamma \vdash \Delta, B} \\
	\infer[\to L]{\Gamma, A\to B \vdash \Delta}{\Gamma \vdash \Delta, A & \Gamma, B \vdash \Delta} && 
	\infer[\to R]{\Gamma \vdash \Delta, A\to B}{\Gamma, A \vdash \Delta, B}
\end{align*}
Negation $\lnot A$ is defined as $A \to \bot$. The names of the structural rules stand, respectively, for `weakening,' `contraction,' and `exchange.' First-order $ \LK$ is the extension of Propositional LK obtained by adding quantifier inferences: 
\begin{align*}
	\infer[\forall L]{\Gamma, \forall x\, A(x) \vdash \Delta}{\Gamma, A(t) \vdash \Delta} && 
	\infer[\exists R]{\Gamma \vdash \Delta, \exists x\, A(x)}{\Gamma \vdash \Delta, A(t)} \\
	\infer[\exists L]{\Gamma, \exists x\, A(x) \vdash \Delta}{\Gamma, A(a) \vdash \Delta} && 
	\infer[\forall R]{\Gamma \vdash \Delta, \forall x\, A(x)}{\Gamma \vdash \Delta, A(a)}
\end{align*}
with the only restriction that in the inferences $\forall R$ and $\exists L$, the variable $a$ is not allowed to appear in the conclusion. In particular, note that every quantifier inference is suitable for every regular $ \LK$-proof. Gentzen's famous \textit{Cut-Elimination Theorem} states that the cut rule is redundant in both propositional and first-order $ \LK$.

\begin{defi}[  $ \LKp$]
The calculus   $ \LKp$ is defined like $ \LK$, except that we instead allow all weak and strong quantifier inferences 
with the proviso that they be suitable for the proof.
\end{defi}

A further weakening of the characteristic-variable conditions gives rise to the notion of weak suitability:
\begin{defi}[Weakly suitable quantifier inference]
A quantifier inference is \textit{weakly suitable for a proof $\pi$} if either it is a weak-quantifier inference or it satisfies substitutability, the side-variable condition, and:
\begin{itemize}
\item (very weak regularity) whenever the characteristic variable is also the characteristic variable of another strong-quantifier inference in $\pi$, then it has the same critical formula.
\end{itemize}
\end{defi}

\begin{defi}[    $ \LKpp$]
The calculus     $ \LKpp$ is the extension of   $ \LKp$ that results from allowing all weakly suitable quantifier inferences.
\end{defi}

\begin{remark}
    Note that any $ \LK$ proof is already $\LKp$ and $\LKpp$ proof, therefore compound axioms can be easily cutfree derived from atomic axioms with linear increase of size.
\end{remark}

It is easy to find examples of sequents that are more easily provable in   $ \LKp$ than in $ \LK$:

\begin{exm} \label{example25}
The sequent 
$$\forall x\, A(x) \rightarrow B \vdash \exists x\, (A(x) \rightarrow B)$$ is cut-free provable in $ \LK$:
\begin{prooftree}
\def\fCenter{\mbox{\ $\vdash$\ }}
\Axiom$A(a) \fCenter A(a)$
\UnaryInf$A(a) \fCenter A(a), B$
\UnaryInf$\fCenter A(a), A(a) \rightarrow B$
\UnaryInf$\fCenter A(a), \exists x\, (A(x) \rightarrow B)$
\UnaryInf$\fCenter \exists x\, (A(x) \rightarrow B), A(a)$
\UnaryInf$\fCenter \exists x\, (A(x) \rightarrow B), \forall x\, A(x)$
\Axiom$B \fCenter B$
\BinaryInf$\forall x\, A(x) \rightarrow B \fCenter \exists x\, (A(x) \rightarrow B), B$
\UnaryInf$\forall x\, A(x) \rightarrow B, A(b) \fCenter \exists x\, (A(x) \rightarrow B), B$
\UnaryInf$\forall x\, A(x) \rightarrow B \fCenter \exists x\, (A(x) \rightarrow B), A(b) \rightarrow B$
\UnaryInf$\forall x\, A(x) \rightarrow B \fCenter \exists x\, (A(x) \rightarrow B), \exists x\, (A(x) \rightarrow B)$
\UnaryInf$\forall x\, A(x) \rightarrow B \fCenter \exists x\, (A(x) \rightarrow B)$
\end{prooftree}
However, one can find a shorter   $ \LKp$-proof:
\begin{prooftree}
\def\fCenter{\mbox{\ $\vdash$\ }}
\Axiom$A(a) \fCenter A(a)$
\UnaryInf$A(a) \fCenter \forall x\, A(x)$
\Axiom$B \fCenter B$
\BinaryInf$A(a), \forall x\, A(x) \rightarrow B \fCenter B$
\UnaryInf$\forall x\, A(x) \rightarrow B, A(a) \fCenter B$
\UnaryInf$\forall x\, A(x) \rightarrow B \fCenter A(a) \rightarrow B$
\UnaryInf$\forall x\, A(x) \rightarrow B \fCenter \exists x\, (A(x) \rightarrow B)$
\end{prooftree}
\qed
\end{exm}

There is an astonishing asymmetry between the proofs for soundness and completeness in the usual proof systems: whereas soundness in Hilbert-style proof systems is simply demonstrated by induction on the length of proofs, completeness is a different matter. Here induction is not sufficient.

\begin{theorem} \label{TheoremCorrectness}
If a sequent is     $ \LKpp$-derivable, then it is already $ \LK$-derivable.
\end{theorem}

\begin{proof}
    
Let $\pi$ be an     $ \LKpp$-proof. Replace every unsound universal quantifier inference by a $\to L$ inference:
\begin{prooftree}
\def\fCenter{\mbox{\ $\vdash$\ }}
\Axiom$\Gamma \fCenter \Delta, A(a)$
\Axiom$\forall x\, A(x) \fCenter \forall x\, A(x)$
\BinaryInf$\Gamma, A(a) \to \forall x\, A(x) \fCenter \Delta, \forall x\, A(x)$
\end{prooftree}
Similarly replace every unsound existential quantifier by an $\to L$ inference
\begin{prooftree}
\def\fCenter{\mbox{\ $\vdash$\ }}
\Axiom$\exists x\, A(x) \fCenter \exists x\, A(x)$
\Axiom$A(a), \Gamma \fCenter \Delta$
\BinaryInf$\Gamma, \exists x\, A(x), \exists x\, A(x) \to A(a) \fCenter \Delta$
\end{prooftree}

By doing this, we obtain a proof of the desired sequent, together with many formulae of the form $A(a) \to \forall x\, A(x)$ or $\exists x\, A(x) \to A(a)$ on the left-hand side. However, we can eliminate each of them by adding an existential quantifier inference and cutting with formulae of the form
$$\vdash \exists y\, \big(A(y) \to \forall x\, A(x)\big)$$
or of the form
$$\vdash \exists y\, \big(\exists x\, A(x) \to A(y)\big),$$
both of which are easily derivable. Note that the existential quantifier inferences can be carried out in a way that is permissible by $ \LK$ because the initial proof satisfied substitutability, weak regularity and the side-variable condition (recall that an $\LKpp $ derivation is already an $\LKp$ derivation).\qed
\end{proof}

\begin{cor}
If a sequent is derivable in   $ \LKp$, then it is already derivable in $ \LK$.
\end{cor}

%\textbf{add examples for each condition}

Each of the three conditions for suitable quantifier inferences have to be fulfilled for $\LKpp$ and first and second condition for $\LKp$.

\begin{enumerate}
\item
(substitutability is not fulfilled) \ \ \ \ \qquad
\AxiomC{$A(a)\vdash A(a)$}
\UnaryInfC{$A(a)\vdash \forall x A(x)$}
\DisplayProof
incorrect result
\item
(side-variable condition is not fulfilled)
\begin{center}
\AxiomC{$A(a,b)\vdash A(a,b)$}
\RightLabel{$*_1$}
\UnaryInfC{$\exists y A(a,y)\vdash A(a, b)$}
\RightLabel{$*_2$}
\UnaryInfC{$\exists y A(a,y)\vdash \forall x A(x, b)$}
\UnaryInfC{$\forall x \exists y A(x,y)\vdash \forall x A(x, b)$}
\UnaryInfC{$\forall x \exists y A(x,y)\vdash \exists y \forall x A(x, y)$}
\DisplayProof
\end{center}
Note that in $*_1$ and $*_2$ we have $b <_\pi a$ and $a <_\pi b$, respectively, incorrect result. %$x$ is a side-variable of $y$ and in $*_2$, $y$ is a side-variable of $x$. Therefore, $x <_\pi y$ and $y <_\pi x$.
\item
(very weak regularity is not fulfilled)
\begin{center}
\AxiomC{$A(a)\vdash A(a)$}
\UnaryInfC{$\exists x A(x)\vdash A(a)$}
\AxiomC{$A(a)\vdash A(a)$}
\UnaryInfC{$A(a)\vdash \forall x A(x)$}
\BinaryInfC{$\exists x A(x)\vdash \forall x A(x)$}
\DisplayProof
\end{center}
incorrect result.
\end{enumerate}

\begin{proposition}%[proof-theoretic observations]
\label{propcontr}

Contrary to $ \LK$  in cut-free proofs in $\LKpp$ weakenings can be moved to the top-most part of the proof and  all contractions can be postponed to the lowest part of the proof with minimal increase of the size. 
\end{proposition}

\begin{proof}
    characteristic variables in $\LKpp$ cannot be blocked by an occurrence outside the scope. Strong quantifiers may be inferred for the same formula from the same characteristic variable (very weak regularity).\qed
\end{proof}

\begin{definition}
    A Herbrand expansion of a sequent with weak quantifiers only is a valid replacement inside-out of all subformulas $\forall x  A(x) $ by some $\bigwedge A(t_i)$ and all subformulas $\exists x  B(x) $ by some $\bigvee B(t_i)$. 
\end{definition}

\begin{proposition}(cf. \cite{DBaazLeitsch} for the prenex case)
    The extraction of a Herbrand expansion from a cut-free proof with weak quantifiers only is elementary.
\end{proposition}

\begin{proof}

Contractions of  $\forall x  A(x) $ ($\exists x  A(x) $) are replaced top-down by conjunctions (disjunctions). This at most doubles the size at any step.   \qed 
\end{proof}

\section{The Relation of $\LKp$ and $\LKpp$ Proofs with Respect to Complexity}

In this article, we use size (number of symbols) as complexity measure of proofs as a length or the number of steps is misleading: even a short proof in this sense might encode a larger proof using high term-complexity.

 $ \LK$, $\LKp$, $\LKpp$ proofs can be easily transformed by addition of cuts into a proof of one of the other formats, with linear increase of complexity. This is not trivial for cut-free proofs.

%Lk++ proofs with cuts can be easily transformed into Lk+ proofs with cuts: note that in the soundness proof the only difference betweem Lk+ and Lk++ proofs ois the potentional contraction of the prostponed formulas before applying cut.

\begin{theorem} \label{th.3}
There is a transformation of a cut-free $\LKpp$-proof $\varphi$ of the end-sequent $\Pi \vdash \Gamma$ into a cut-free $\LKp$-proof of $\Pi \vdash \Gamma$ with an elementary bound in the complexity of $\varphi$.
\end{theorem}
\begin{proof}
Let $\varphi$ be a cut-free $\LKpp$-derivation with end-sequent $\Pi \vdash \Gamma$ and characteristic variables $a_1, \ldots, a_n$ of strong quantifier inferences.
We replace the old inferences stepwise by inferences with the stronger regularity conditions of $\LKp$. These new variables are of the form 
$$x_{f(a_{i_1}, \ldots, a_{i_k})},$$
where $a_{i_1}, \ldots, a_{i_k}$ are the old variables still present at the time of the transformation.
Consider all inferences $\forall_r$ of the form
\begin{prooftree}
\AxiomC{$\Pi_j \vdash \Gamma_j, A(a)$}
\UnaryInfC{$\Pi_j \vdash \Gamma_j, \forall x A(x)$}
\end{prooftree}
where this inference is an innermost inference not yet changed. Extract a proof $\varphi'(a)$ of 
$$\Pi \vdash \Gamma, A(a)$$
by deleting all inferences working on $A(a)$ and resulting formulas. Locate all places in the proof where $\forall x A(x)$ is inferred from $A(a)$. 
Let $i$ be such a position. We replace the sub-proof leading to the premise 
$$\Pi \vdash \Gamma, A(a)$$
of this inference by $\varphi'(x_{f_i(a_{i_1}, \ldots, a_{i_k})})$ of 
$$\Pi(x_{f_i(a_{i_1}, \ldots, a_{i_k})}) \vdash \Gamma(x_{f_i(a_{i_1}, \ldots, a_{i_k})}), A(x_{f_i(a_{i_1}, \ldots, a_{i_k})}).$$
Then we replace $a$ in the indices of new variables everywhere by $f_i$. By weakening we obtain 
$$\Pi, \Pi(x_{f_i(a_{i_1}, \ldots, a_{i_k})}) \vdash \Gamma(x_{f_i(a_{i_1}, \ldots, a_{i_k})}), \Gamma, A(x_{f_i(a_{i_1}, \ldots, a_{i_k})})$$
and continue with the proof. ($\exists$ left inferences are dual.) \\[1ex]
Note that by replacing $a$ in the indices inferences already transformed to regular inferences remain regular. 
The side-variable condition is extended correspondingly. Note that the final proof of the original end-sequent after contraction which obviously fulfils the strong regularity condition of $\LKp$ also fulfils the non-occurrence of strong inference variables in the end-sequent and the side-variable conditions are acyclic, as $x_f < x_g$ only if $g$ is a proper subterm of $f$.\\[1ex]
Finally, we replace all occurrences of $a$ (which are no characteristic variables of any inference anymore) by a fixed constant $c$.\\[1ex]
The transformation is exponential in the number of variables w.r.t. the height if weakenings and contractions are not counted. It is therefore elementary in the complexity.\qed
\end{proof}

\begin{example} \label{ex.lkpp}
Consider the following $\LKpp$-proof $\varphi$:

\begin{prooftree}
\AxiomC{$A(a,a,b) \vdash A(a,a,b)$}
\UnaryInfC{$A(a,a,b) \vdash \forall z A(a,a,z)$}
\UnaryInfC{$A(a,a,b) \vdash \exists y \forall z A(a,y,z)$}
\UnaryInfC{$A(a,a,b) \vdash \forall x \exists y \forall z A(x,y,z)$}
\UnaryInfC{$C, A(a,a,b) \vdash \forall x \exists y \forall z A(x,y,z)$}

\AxiomC{$A(a,a,b) \vdash A(a,a,b)$}
\UnaryInfC{$A(a,a,b) \vdash \forall z A(a,a,z)$}
\UnaryInfC{$A(a,a,b) \vdash \exists y \forall z A(a,y,z)$}
\UnaryInfC{$A(a,a,b) \vdash \forall x \exists y \forall z A(x,y,z)$}
\UnaryInfC{$D, A(a,a,b) \vdash \forall x \exists y \forall z A(x,y,z)$}

\BinaryInfC{$C \lor D, A(a,a,b) \vdash \forall x \exists y \forall z A(x,y,z)$}
\UnaryInfC{$C \lor D, \forall y A(a,a,y) \vdash \forall x \exists y \forall z A(x,y,z)$}
\UnaryInfC{$C \lor D, \forall x \forall y A(x,x,y) \vdash \forall x \exists y \forall z A(x,y,z)$}
\end{prooftree}
with $b <_{\varphi, \LK} a$.\\[1ex]
In the following we will use the abbreviations $\Pi$ for $C \lor D, \forall x \forall y A(x,x,y)$ and $\Gamma$ for $\forall x \exists y \forall z A(x,y,z)$ (our end-sequent is $\Pi \vdash \Gamma$). \\[1ex]
$\varphi(b)$ is:
\begin{prooftree}
\AxiomC{$A(a,a,b) \vdash A(a,a,b)$}
\UnaryInfC{$C, A(a,a,b) \vdash A(a,a,b)$}
\AxiomC{$A(a,a,b) \vdash A(a,a,b)$}
\UnaryInfC{$D, A(a,a,b) \vdash A(a,a,b)$}
\BinaryInfC{$C \lor D, A(a,a,b) \vdash A(a,a,b)$}
\UnaryInfC{$C \lor D, \forall y A(a,a,y) \vdash A(a,a,b)$}
\UnaryInfC{$C \lor D, \forall x \forall y A(x,x,y) \vdash A(a,a,b)$}
\UnaryInfC{$\Pi \vdash \Gamma, A(a,a,b)$}
\end{prooftree}
The new proof is obtain from $\varphi(x_{f(a)})$ and $\varphi(x_{g(a)})$:
\begin{prooftree}
\AxiomC{$(\varphi(x_{f(a)}))$}
\UnaryInfC{$\Pi \vdash \Gamma, A(a,a, x_{f(a)})$}
\UnaryInfC{$\Pi \vdash \Gamma, \forall z A(a,a, z)$}
\UnaryInfC{$A(a,a,b), \Pi \vdash \Gamma, \forall z A(a,a, z)$}
\AxiomC{$(\varphi(x_{g(a)}))$}
\UnaryInfC{$\Pi \vdash \Gamma, A(a,a, x_{g(a)})$}
\UnaryInfC{$\Pi \vdash \Gamma, \forall z A(a,a, z)$}
\UnaryInfC{$A(a,a,b), \Pi \vdash \Gamma, \forall z A(a,a, z)$}
\doubleLine
\BinaryInfC{$\Pi, \Pi \vdash \Gamma, \Gamma$}
\UnaryInfC{$\Pi \vdash \Gamma$}
\end{prooftree}
where $x_{f(a)} < a$ and $x_{g(a)} < a$.\\[1ex]
$\varphi'(a, x_{f(a)}, x_{g(a)})$ is:
\begin{prooftree}
\AxiomC{$(\varphi(x_{f(a)}))$}
\UnaryInfC{$\Pi \vdash \Gamma, A(a,a, x_{f(a)})$}
\UnaryInfC{$\Pi \vdash \Gamma, \forall z A(a,a, z)$}
\UnaryInfC{$A(a,a,b), \Pi \vdash \Gamma, \forall z A(a,a, z)$}
\UnaryInfC{$A(a,a,b), \Pi \vdash \Gamma, \exists y \forall z A(a,y,z)$}
\UnaryInfC{$C, A(a,a,b), \Pi \vdash \Gamma, \exists y \forall z A(a,y,z)$}
\AxiomC{$(\varphi(x_{g(a)}))$}
\UnaryInfC{$\Pi \vdash \Gamma, A(a,a, x_{g(a)})$}
\UnaryInfC{$\Pi \vdash \Gamma, \forall z A(a,a, z)$}
\UnaryInfC{$A(a,a,b), \Pi \vdash \Gamma, \forall z A(a,a, z)$}
\UnaryInfC{$A(a,a,b), \Pi \vdash \Gamma, \exists y \forall z A(a,y,z)$}
\UnaryInfC{$D, A(a,a,b), \Pi \vdash \Gamma, \exists y \forall z A(a,y,z)$}
\BinaryInfC{$C \lor D, A(a,a,b), \Pi \vdash \Gamma, \exists y \forall z A(a,y,z)$}
\UnaryInfC{$C \lor D, \forall y A(a,a,y), \Pi \vdash \Gamma, \exists y \forall z A(a,y,z)$}
\UnaryInfC{$C \lor D, \forall x \forall y A(x,x,y), \Pi \vdash \Gamma, \exists y \forall z A(a,y,z)$}
\end{prooftree}
Finally,
\begin{prooftree}
\AxiomC{$(\varphi'(x_d, x_{f(d)}, x_{g(d)}))$}
\UnaryInfC{$\Pi \vdash \Gamma, \exists y \exists z A(x_d,y,z)$}
\UnaryInfC{$\Pi \vdash \Gamma, \forall x \exists y \exists z A(x,y,z)$}
\UnaryInfC{$A(a,a,b), \Pi \vdash \Gamma, \forall x \exists y \exists z A(x,y,z)$}
\UnaryInfC{$C, A(a,a,b), \Pi \vdash \Gamma, \forall x \exists y \exists z A(x,y,z)$}

\AxiomC{$(\varphi'(x_e, x_{f(e)}, x_{g(e)}))$}
\UnaryInfC{$\Pi \vdash \Gamma, \exists y \exists z A(x_e,y,z)$}
\UnaryInfC{$\Pi \vdash \Gamma, \forall x \exists y \exists z A(x,y,z)$}
\UnaryInfC{$A(a,a,b), \Pi \vdash \Gamma, \forall x \exists y \exists z A(x,y,z)$}
\UnaryInfC{$D, A(a,a,b), \Pi \vdash \Gamma, \forall x \exists y \exists z A(x,y,z)$}

\BinaryInfC{$C \lor D, A(a,a,b), \Pi \vdash \Gamma, \forall x \exists y \exists z A(x,y,z)$}
\UnaryInfC{$C \lor D, \forall y A(a,a,y), \Pi \vdash \Gamma, \forall x \exists y \exists z A(x,y,z)$}
\UnaryInfC{$C \lor D, \forall x \forall y A(x,x,y), \Pi \vdash \Gamma, \forall x \exists y \exists z A(x,y,z)$}
\doubleLine
\UnaryInfC{$\Pi \vdash \Gamma$}
\end{prooftree}
where $x_{f(e)} < x_e, x_{g(e)} < x_e, x_{f(d)} < x_d, x_{g(d)} < x_d$.
\end{example}

\section{Quantifier Shifts in the Cut-free Contexts}\label{quant}

The aboriginal proof strategy of quantifier shifts occurring e.g. in prenexation in classic logic is orthogonal to the concept of cut-free $ \LK$ proofs as shown below.   $ \LKp$ however allows to simulate quantifier shifts with low increase of complexity. For this reason, we introduce a calculus for quantifier shifts

\begin{defi}
The calculus $\lkqs$ is obtained by extending $ \LK$ with the following rules:
\begin{align*}
\infer[]{\Gamma, \kappa[Q'x\, (A \lhd B)] \vdash \Delta}{\Gamma, \kappa[Qx\, A \lhd B] \vdash \Delta} &&
\infer[]{\Gamma, \kappa[Q'x\, (A \lhd B)] \vdash \Delta}{\Gamma, \kappa[A \lhd Qx\,B] \vdash \Delta}\\
\infer[]{\Gamma \vdash \Delta, \kappa[Q'x\, (A \lhd B)]}{\Gamma \vdash \Delta, \kappa[Qx\, A \lhd B]} &&
\infer[]{\Gamma \vdash \Delta, \kappa[Q'x\, (A \lhd B)]}{\Gamma \vdash \Delta, \kappa[A \lhd Qx\, B]}
\end{align*}
where $\kappa[\cdot]$ is a context, $\lhd \in \{\wedge, \vee, \to\}$ and $Q' = Q$, except if $\lhd$ is $\to$ and $Q$ is taken from the antecedent, in which case $Q'$ is opposite. We refer to these rules as \textit{deep quantifier shifts}. 
\end{defi}

Here, recall that, as per our conventions, the syntax of first-order logic includes separate symbols for free and bound variables.

\begin{theorem} (cf. Aguilera, Baaz\cite{DBLP:journals/jsyml/AguileraB19}) \label{PropLKQS}
Cut-free   $ \LKp$ simulates cut-free $\lkqs$ double-exponentially, i.e., every $\lkqs$-provable sequent is   $ \LKp$-provable and there is a double exponential function that bounds the length of the least cut-free   $ \LKp$-proof of a   $ \LKp$-provable sequent in terms of its least cut-free $\lkqs$-proof.
\end{theorem}

\begin{proof}
    
Let $\pi$ be a cut-free $\lkqs$-proof of size $n$. We transform it into a cut-free   $ \LKp$-proof. Assume by induction that there is only one application of a quantifier shift and this is the last inference. We proceed by cases. For simplicity, we assume $\kappa$ does not change the polarity of the quantifiers. We also assume $\lhd$ is $\to$; the other connectives are treated similarly. \\

\noindent\textsc{Case I.} The last inference is: 
\begin{align*}
\infer[]{\Gamma, \kappa[\exists x\, (A(x) \to B)] \vdash \Delta}{\Gamma, \kappa[\forall x\, A(x) \to B] \vdash \Delta} 
\end{align*}
so that the proof has the following structure:
\begin{prooftree}
\def\fCenter{\mbox{\ $\vdash$\ }}
\AxiomC{$\vdots\sigma_2$}
\UnaryInf$\Gamma'' \fCenter \Delta'', A(a)$
\RightLabel{($*$)}
\UnaryInf$\Gamma'' \fCenter \Delta'', \forall x\, A(x)$
\UnaryInfC{$\vdots\sigma_1$}
\UnaryInf$\Gamma' \fCenter \Delta', \forall x\, A(x)$
\AxiomC{$\vdots\sigma_3$}
\UnaryInf$\Gamma', B\fCenter \Delta'$
\BinaryInf$\Gamma', \forall x\, A(x) \to B \fCenter \Delta'$
\UnaryInfC{$\vdots\sigma_0$}
\UnaryInf$\Gamma, \kappa[\forall x\, A(x) \to B] \fCenter \Delta$
\UnaryInf$\Gamma, \kappa[\exists x\, (A(x) \to B)] \fCenter \Delta$
\end{prooftree}
where the $\sigma_i$ denote subproofs. The subproof $\sigma_0$ might split into several branches not indicated in the diagram, each of which could potentially include its own copy of $\forall x\, A(x)$ being inferred and its own subproof $\sigma_3'$ with conclusion of the form $\Pi, B \vdash \Lambda$. Each branch can be dealt with the same way and so we will only focus on the indicated parts of $\pi$.

Similarly, the subproof $\sigma_1$ might split into several branches not indicated in the diagram, each of which could potentially include its own copy of $\forall x\, A(x)$ being inferred. This we need to bear in mind. Note that each of those copies of $\forall x\, A(x)$ is necessarily inferred from a different characteristic variable. In fact, the variable $a$ does not appear below $(*)$, by the regularity of $\lkqs$.

We modify the proof. Our approach is as follows: we would like to merge the subproofs $\sigma_1$ and $\sigma_2$ simply by postponing the inference $(*)$ as follows:
\begin{prooftree}
\def\fCenter{\mbox{\ $\vdash$\ }}
\AxiomC{$\vdots\sigma_2 + \sigma_1$}
\UnaryInf$\Gamma' \fCenter \Delta', A(a)$
\AxiomC{$\vdots\sigma_3$}
\UnaryInf$\Gamma', B\fCenter \Delta'$
\BinaryInf$\Gamma', A(a) \to B \fCenter \Delta'$
\RightLabel{($**$)}
\UnaryInf$\Gamma', \exists x\, ( A(x) \to B) \fCenter \Delta'$
\UnaryInfC{$\vdots\sigma_0$}
\UnaryInf$\Gamma, \kappa[\exists x\, (A(x) \to B)] \fCenter \Delta$
\end{prooftree}
The problem that might arise is that some occurrence of $\forall x\, A(x)$ that would be contracted in $\sigma_1$ with the indicated occurrence is unable to be contracted. We describe how to deal with each of them.

Notice that the problematic occurrence of $\forall x\, A(x)$ must originate from a strong-quantifier inference. Omit that inference and drag the unquantified formula $A(b)$ until after $(**)$, so that at that point we have a derivation of 
$$\Gamma'_0, A(b), \Gamma'_1, \exists x\, ( A(x) \to B), \vdash \Delta',$$
where $\Gamma' = \Gamma'_0, \Gamma'_1$. Add to that some exchanges to obtain:
$$\Gamma', \exists x\, ( A(x) \to B), A(b) \vdash \Delta',$$
and a subproof
\begin{align} \label{eqqsCaseI}
\infer[]{\Gamma', B \vdash \Delta'}{\vdots\sigma_3}
\end{align}
in order to infer 
$$\Gamma', \exists x\, ( A(x) \to B), A(b) \to B \vdash \Delta'.$$
with an application of $\exists L$ and a contraction, we are left again with 
$$\Gamma', \exists x\, ( A(x) \to B) \vdash \Delta',$$
to which we can apply $\sigma_0$ as desired. The proof grows by adding a copy of \eqref{eqqsCaseI}---of size at most $n$---for each problematic occurrence of $\forall x\, A(x)$ (of which there are at most $n$). Hence, it grows quadratically.\\

\noindent\textsc{Case II.} The last inference is:
\begin{align*}
\infer[]{\Gamma, \kappa[\forall x\, (A(x) \to B)] \vdash \Delta}{\Gamma, \kappa[\exists x\, A(x) \to B] \vdash \Delta} 
\end{align*}
The proof has the following structure:
\begin{prooftree}
\def\fCenter{\mbox{\ $\vdash$\ }}
\AxiomC{$\vdots\sigma_2$}
\UnaryInf$\Gamma'' \fCenter \Delta'', A(t)$
\UnaryInf$\Gamma'' \fCenter \Delta'', \exists x\, A(x)$
\UnaryInfC{$\vdots\sigma_1$}
\UnaryInf$\Gamma' \fCenter \Delta', \exists x\, A(x)$
\AxiomC{$\vdots\sigma_3$}
\UnaryInf$\Gamma', B\fCenter \Delta'$
\BinaryInf$\Gamma', \exists x\, A(x) \to B \fCenter \Delta'$
\UnaryInfC{$\vdots\sigma_0$}
\UnaryInf$\Gamma, \kappa[\exists x\, A(x) \to B] \fCenter \Delta$
\UnaryInf$\Gamma, \kappa[\forall x\, (A(x) \to B)] \fCenter \Delta$
\end{prooftree}
As in Case I, the subproof $\sigma_0$ could branch off into several subtrees each of which is taken care of in the same way, and so we only focus on the indicated parts of $\pi$.
We would like to merge the subproofs $\sigma_1$ and $\sigma_2$ as follows:
\begin{prooftree}
\def\fCenter{\mbox{\ $\vdash$\ }}
\AxiomC{$\vdots\sigma_2 + \sigma_1$}
\UnaryInf$\Gamma' \fCenter \Delta', A(t)$
\AxiomC{$\vdots\sigma_3$}
\UnaryInf$\Gamma', B\fCenter \Delta'$
\BinaryInf$\Gamma', A(t) \to B \fCenter \Delta'$
\UnaryInf$\Gamma', \forall x\, ( A(x) \to B) \fCenter \Delta'$
\UnaryInfC{$\vdots\sigma_0$}
\UnaryInf$\Gamma, \kappa[\forall x\, (A(x) \to B)] \fCenter \Delta$
\end{prooftree}
As before, we face the problem of circumventing a contraction of (possibly several occurrences of) $\exists x\, A(x)$ in $\sigma_1$ and solve it in the same way. The proof grows quadratically again. We note the following: it might happen that the term $t$ contains some free variable $a$. Hence, the proof of this case does not necessarily go through in $ \LK$, as $a$ could be the characteristic variable of a strong-quantifier inference in $\sigma_1$. Note that the side-variable condition is satisfied. \\

\noindent\textsc{Case III.} The last inference is: 
\begin{align*}
\infer[]{\Gamma, \kappa\big[\forall x\, (A \to B(x))\big] \vdash \Delta}{\Gamma, \kappa\big[A \to \forall x\,B(x)\big] \vdash \Delta} 
\end{align*}
This is analogous to \textsc{Case II.}\\

\noindent\textsc{Case IV.} The last inference is: 
\begin{align*}
\infer[]{\Gamma, \kappa\big[\exists x\, (A \to B(x))\big] \vdash \Delta}{\Gamma, \kappa\big[A \to \exists x\,B(x)\big] \vdash \Delta} 
\end{align*}
This is analogous to \textsc{Case I.}\\

\noindent\textsc{Case V.} The last inference is: 
\begin{align*}
\infer[]{\Gamma \vdash \Delta, \kappa\big[\exists x\, (A(x) \to B)\big]}
{\Gamma \vdash \Delta, \kappa\big[\forall x\, A(x) \to B\big]} 
\end{align*}
The proof has the following form, modulo qualifications as in Cases I and II:
\begin{prooftree}
\def\fCenter{\mbox{\ $\vdash$\ }}
\AxiomC{$\vdots\sigma_2$}
\UnaryInf$\Gamma'', A(t) \fCenter \Delta''$
\RightLabel{($*$)}
\UnaryInf$\Gamma'', \forall x\, A(x) \fCenter \Delta''$
\UnaryInfC{$\vdots\sigma_1$}
\UnaryInf$\Gamma', \forall x\, A(x) \fCenter \Delta', B$
\UnaryInf$\Gamma' \fCenter \Delta', \forall x\, A(x) \to B$
\UnaryInfC{$\vdots\sigma_0$}
\UnaryInf$\Gamma \fCenter \Delta, \kappa\big[\forall x\, A(x) \to B\big]$
\UnaryInf$\Gamma \fCenter \Delta, \kappa\big[\exists x\, (A(x) \to B)\big]$
\end{prooftree}
We would like to postpone the inference $(*)$, thus merging $\sigma_1$ and $\sigma_2$:
\begin{prooftree}
\def\fCenter{\mbox{\ $\vdash$\ }}
\AxiomC{$\vdots\sigma_2 + \sigma_1$}
\UnaryInf$\Gamma', A(t) \fCenter \Delta', B$
\RightLabel{($**$)}
\UnaryInf$\Gamma' \fCenter \Delta', A(t) \to B$
\UnaryInf$\Gamma' \fCenter \Delta', \exists x\, ( A(x) \to B)$
\UnaryInfC{$\vdots\sigma_0$}
\UnaryInf$\Gamma \fCenter \Delta, \kappa[\exists x\, (A(x) \to B)]$
\end{prooftree}
However, we might have---similarly to the previous cases---an occurrence of $\forall x\, A(x)$ that would be contracted on the left-hand side as part of $\sigma_1$. For simplicity, assume there is only one. This occurrence is inferred from a formula $A(s)$. Postpone this occurrence until after $(**)$ so that at that point we have a derivation of 
$$\Gamma'_0, A(s), \Gamma'_1 \vdash \Delta', \exists x\, (A(x) \to B),$$
where $\Gamma' = \Gamma'_0, \Gamma'_1$. Add to that some exchanges to obtain:
$$\Gamma', A(s) \vdash \Delta', \exists x\, (A(x) \to B),$$
and a weakening:
\begin{align} \label{eqqsCaseV}
\infer[WR]{\Gamma', A(s) \vdash \Delta', \exists x\, (A(x) \to B), B}
{\Gamma', A(s) \vdash \Delta', \exists x\, (A(x) \to B)}
\end{align}
in order to infer 
$$\Gamma' \vdash \Delta', \exists x\, (A(x) \to B), A(s) \to B.$$
with an application of $\exists R$ and a contraction, we are left again with 
$$\Gamma' \vdash \Delta', \exists x\, (A(x) \to B),$$
to which we can apply the rest of $\sigma_0$ as desired. In this case the proof only grows linearly. As in \textsc{Case II}, the proof of this case would not necessarily work for $ \LK$, but it is correct from the point of view of   $ \LKp$, as follows from the regularity of $\lkqs$ and the fact that the conclusion of the proof is $\pi$.\\

\noindent\textsc{Case VI.} The last inference is: 
\begin{align*}
\infer[]{\Gamma \vdash \Delta, \kappa\big[\forall x\, (A(x) \to B)\big]}
{\Gamma \vdash \Delta, \kappa\big[\exists x\, A(x) \to B\big]} 
\end{align*}
The proof has the following form, modulo qualifications as in Cases I and II:
\begin{prooftree}
\def\fCenter{\mbox{\ $\vdash$\ }}
\AxiomC{$\vdots\sigma_2$}
\UnaryInf$\Gamma'', A(a) \fCenter \Delta''$
\RightLabel{($*$)}
\UnaryInf$\Gamma'', \exists x\, A(x) \fCenter \Delta''$
\UnaryInfC{$\vdots\sigma_1$}
\UnaryInf$\Gamma', \exists x\, A(x) \fCenter \Delta', B$
\UnaryInf$\Gamma' \fCenter \Delta', \exists x\, A(x) \to B$
\UnaryInfC{$\vdots\sigma_0$}
\UnaryInf$\Gamma \fCenter \Delta, \kappa\big[\exists x\, A(x) \to B\big]$
\UnaryInf$\Gamma \fCenter \Delta, \kappa\big[\forall x\, (A(x) \to B)\big]$
\end{prooftree}
As before, we would like to postpone the inference $(*)$, thus merging $\sigma_1$ and $\sigma_2$:
\begin{prooftree}
\def\fCenter{\mbox{\ $\vdash$\ }}
\AxiomC{$\vdots\sigma_2 + \sigma_1$}
\UnaryInf$\Gamma', A(a) \fCenter \Delta', B$
\UnaryInf$\Gamma' \fCenter \Delta', A(a) \to B$
\UnaryInf$\Gamma' \fCenter \Delta', \forall x\, ( A(x) \to B)$
\UnaryInfC{$\vdots\sigma_0$}
\UnaryInf$\Gamma \fCenter \Delta, \kappa[\forall x\, (A(x) \to B)]$
\end{prooftree}
We deal with contractions of $\exists x\, A(x)$ on the left-hand side as in the previous case. The proof grows linearly.\\

\noindent\textsc{Case VII.} The last inference is: 
\begin{align*}
\infer[]{\Gamma \vdash \Delta, \kappa\big[\forall x\, (A \to B(x))\big]}
{\Gamma \vdash \Delta, \kappa\big[A \to \forall x\,B(x)\big]} 
\end{align*}
This is analogous to \textsc{Case VI.}\\

\noindent\textsc{Case VIII.} The last inference is: 
\begin{align*}
\infer[]{\Gamma \vdash \Delta, \kappa\big[\exists x\, (A \to B(x))\big]}{\Gamma \vdash \Delta, \kappa\big[A \to \exists x\,B(x)\big]} 
\end{align*}
This is analogous to \textsc{Case V.}\\

Hence, we have dealt with all the cases. Note that for each quantifier shifted, the proof grows at most quadratically. Since there are at most $n$ deep quantifier shifts, the size of the resulting proof is bounded by 
$$n^{2^n} = 2^{2^n\cdot \log(n)} \approx 2^{2^n}$$ and this is elementary.\qed
\end{proof}

\section{Standard Skolemization and Andrews Skolemization}\label{skol}

Proof-theoretic Skolemization relates to the replacement of characteristic variables in strong quantifier inferences by Skolem terms and preserves validity. In structural Skolemization, i.e., standard case, the Skolem functions depend on all weak quantifiers in whose scope the transformed subformula occurs.

\begin{definition}\label{Def: St Skolem} 
 Let $A$ be a closed first-order formula. Whenever $A$ does not contain strong quantifiers, we define its \emph{structural Skolem form} as $A^S := A$. Suppose now that $A$ contains strong quantifiers. Let $Q \in \{\forall, \exists\}$ and $(Q y)$ be the first strong quantifier occurring in $A$. If $(Q y)$ is not in the scope of weak quantifiers, then its structural Skolem form is
$$ A^S := (A_{-(Q y)}\{y \leftarrow c\})^S,$$
where $A_{-(Q y)}$ is the formula $A$ after omission of $(Q y)$ and $c$ is a new constant symbol not occurring in $A$.
If $(Q y)$ is in the scope of the weak quantifiers $(Q_1 x_1) \ldots (Q_n x_n)$, then its structural Skolemization is
$$ A^S:= (A_{-(Q y)}\{y \leftarrow f(x_1, \ldots , x_n)\})^S,$$
where $f$ is a fresh (Skolem) function symbol and does not occur in $A$. 
Given $F = \bigwedge \Gamma \to \bigvee \Delta$ and $F^S= \bigwedge \Pi \to \bigvee \Lambda$, define the structural Skolemization of the sequent $\Gamma\vdash \Delta$ as $$(\Gamma\vdash \Delta)^S := \Pi\vdash \Lambda.$$
%the Skolemization of a first-order formula A is defined by replacing every strongly quantified variable y with a new function symbol fA y (x1,...,xn), where x1,...,xn are the weakly quantified variables such that Qy appears in the scope of their quantifiers, and removing the quantifier Qy.
\end{definition}

This substitutions increases the size of cut-free proofs in $Lk$ polynomially. In $\LKp$ and $\LKpp$ the standard Skolemization for cut-free proofs in Baaz, Leitsch\cite{baaz2011methods} is not possible (as seen later), i.e., cuts have to be added as in the following example.

\begin{example}\label{example}

%\begin{center}
%\begin{prooftree}
%\def\fCenter{\mbox{\ $\vdash$\ }}
\AxiomC{$A(t,f(t))\limp A(t,f(t))$}
\UnaryInfC{ $\forall y A(t,y)\limp A(t,f(t)) $}
\UnaryInfC{$\vdots$}
\UnaryInfC{$\Pi\limp \Gamma, \forall y A(t,y)$  $\forall y A(t,y),\Pi\limp \Gamma, A(t,f(t))  $}
\UnaryInfC{$\Pi\limp \Gamma, A(t,f(t))$}
\UnaryInfC{$\Pi\limp \Gamma, \exists x A(x,f(x))$}
%\end{prooftree}
\DisplayProof
%\end{center}
\end{example}

\begin{definition}\label{Def: And Skolem} 
For a closed first-order formula $A$, \emph{Andrews Skolemization} $A^\SA$ is defined similar to Definition \ref{Def: St Skolem}.
The only difference is that the introduced Skolem functions in Andrews method do not depend on the weak quantifiers $(Q_1 x_1) \ldots (Q_n x_n)$ that dominate the strong quantifier $(Q x)$, but depend on the subset of $\{x_1, \ldots, x_n\}$ appearing (free) in the subformula dominated by $(Q x)$.
\end{definition}
Andrews Skolemization may allow shorter Herbrand sequent.

\begin{example}
Consider the sequent 
$$S \colon \forall x \forall y (\exists z P(x,z) \lor Q(y,x)) \vdash.$$
Then its structural Skolemization is
$$\forall x \forall y (P(x,f(x,y)) \lor Q(y,x)) \vdash$$
and its Andrews Skolemization is
$$\forall x \forall y (P(x,g(x)) \lor Q(y,x)) \vdash.$$
\end{example}

\begin{remark}
Herbrand disjunctions are generally shortened, cf the formula
$$ \exists y (\exists x A(x) \to A(y)).$$
The structural Skolemization of this formula is $ \exists y (A(f(y)) \to A(y))$, with shortest Herbrand disjunction
$$ A(f(f(a)) \to A(f(a)) \lor A(f(a)) \to A(a).$$
The Andrews Skolemization of the above formula is $\exists y (A(c) \to A(y))$, with shortest Herbrand disjunction
$$ A(c) \to A(c).$$
\end{remark}

\begin{remark}
   In the satisfiability sense, Structural Skolemization is obtained by addition of 
    \[ 
\forall x_1\dots x_n(\exists y A(x_1\dots x_n, y)\rightarrow \exists y A(x_1\dots x_n, y))    
    \]
for corresponding formulas and subformulas.
Skolemize them using the axiom of choice   \[ 
\forall x_1\dots x_n(\exists y A(x_1\dots x_n, y)\rightarrow  A(x_1\dots x_n, f(x_1\dots x_n)))   
    \]
    and replace strong quantifiers of the original formulas by Skolem terms.
    The difference is that it verifies in addition
    \[ 
    \bigwedge x_i=x'_i\rightarrow f(x_1\dots x_n)=f(x'_1\dots x'_n).
    \]

    This is not included in the proof-theoretic Skolemization as the proof-theoretic deSkolemization methods fail (note that Skolemization in this satifiability context is dual to Skolemization in proof-theoretic context. In the first case, the challenge is Skolemization and in the second case - deSkolemization).
\end{remark}

\begin{remark}
    Andrews Skolemization has been originally introduced by the equal satisfiability of \[A(\exists x B(x_1\dots x_n))\]
    and 
\[A(P(x_1\dots x_n))\wedge \forall x_1\dots x_n(\exists x B( x_1\dots x_n)\leftrightarrow P(x_1\dots x_n)) \]
$x_1\dots x_n$ bound by weak quantifiers outside $B$ and occuring in $B$, $P$ new predicate. Skolemize the right component with structural Skolemization and use the equivalence to obtain \[A(B(f(x_1\dots x_n),x_1\dots x_n)).\]
\end{remark}

Structural Skolemization can be elementarily deSkolemized in cut-free proofs in $ \LK$ by Baaz, Hetzel, Weller \cite{Hetzl}. (This deSkolemization is double exponential contrary to the deSkolemization in the second epsilon theorem which is exponentially bounded but relates only to prenex formulas). In $\LKpp$ the following deSkolemization is possible:
\begin{itemize}
    \item Move all contractions to the bottom according to the Proposition \ref{propcontr}.
    \item Replace the Skolem terms by variables indexed by the Skolem terms.
     \item Replace the weakening by the adequate formulas with strong quantifiers.
     \item  Order the original Skolem terms from large size to small size.
     \item Infer weak quantifiers step by step at any stage wherever it is possible.
      \item Infer strong quantifiers step by step at the top-most possibility for the characteristic variables $x_s$ where $s$ is a maximal Skolem term indexing a still free characteristic variable.
      \item Repeat this procedure until all free characteristic variables are exhausted.
            \end{itemize}

The soundness is guaranteed by the dependencies of the Skolem functions in the Skolem terms (the side-variable condition is satisfied and very weak regularity enables the top-most inferences).
This is an extended form of the second epsilon theorem allowing for the deSkolemization of all sorts of strong quantifiers also with infix occurrences. The price is, however, that it does not lead immediately to a cut-free proof to $ \LK$ format.

Structural Skolemization is elementarily associated with cut-free proofs in $ \LK$, it is not elementarily associated with cut-free $\LKpp$ proofs, i.e., there is no elementary translation of cut-free $\LKpp$ proofs to the cut-free structural Skolemized form (see the next section). Such an elementary translation exists, however, for Andrews Skolemization.

\section{ Non-elementary Speed-ups of Cut-free $\LKp$ and $\LKpp$ Proofs Relative to Cut-free $\LK$-Proofs  }\label{speed}

%\textbf{definition of elementary.}
 The arguments in this section are based on the well-known result of Orevkov\cite{orevkov1982lower}.

\begin{theorem}\label{th2}
   There is a specific family of sequents $\{S_i \}_{i<\omega}$ and specific $ \LK$-proofs such that they have the following properties:
    \begin{enumerate}
        \item the size of $S_i$ is polynomial in $i$;
        \item there is no bound on the size of their smallest cut-free $ \LK$-proofs that is elementary in $i$;
        \item the size of these proofs (with cuts), however, is polynomially bounded in $i$.
    \end{enumerate}
 
\end{theorem}

\begin{corollary}
    Each worst-case sequence as formulated in Theorem \ref{th2} generates a worst-case
sequence, where the end-sequents contain weak quantifiers only.
\end{corollary}

\begin{proof}
    Skolemize the proofs in the worst-case sequence according to Section \ref{skol} and deSkolemize according to \cite{Hetzl}.\qed
\end{proof}

\begin{remark}
   % In the forth-coming proofs any worst case sequence can be used as a distinction between elementary and non-elementary.
    We could use any worst case sequence, e.g., that of Statman \cite{statman1979lower} or that of Pudl\'ak \cite{Pavel}. In Statman's case, the formalization of Statman result in $LK$ should be used\cite{baaz2011methods}.
\end{remark}

\begin{proposition}\label{size}
For every proof with cuts of size $\sigma$ 
 there is a cut-free proof of size at most $2\sigma^2)$ where $\forall \bar{x}$ closed is added to the end-sequent.
    \end{proposition}
 \begin{proof}
     Replace cuts by implication left inferences, and close under free variables.\qed
 \end{proof}

This makes it possible to refine the worst-case sequent of Theorem \ref{th2}.

\begin{theorem}\label{th3}
   There is a specific family of sequents $\{S_i \}_{i<\omega}$ and specific $ \LK$-proofs such that they have the following properties:
    \begin{enumerate}
        \item the size of $S_i$ is elementary in $i$, and  $S_i$ contains only weak quantifiers;
        \item there is no bound on the size of their smallest cut-free $ \LK$-proofs that is elementary in $i$;
        \item the size of these proofs (with cuts), however, is elementary bounded in $i$;
        \item all cut formulas are closed and prenex.
    \end{enumerate}
 
\end{theorem}

\begin{proof}
    We use the cut-free proofs of Proposition \ref{th2}. We derive corresponding sequent  $\forall x A    \rightarrow A$ and cut them out. Now we make the cuts prenex. This is done by stepwise deriveing outside the shift of quantifiers and cutting the infix inference of quantifier as in Baaz, Leisch \cite{DBLP:journals/apal/BaazL99}.    
    Similar for all other connectives (for details see \cite{DBLP:journals/apal/BaazL99}).
    Therefore the size of the proofs in the worst case sequence remains elementary in $i$.
    \qed
\end{proof}

In Section \ref{quant}, we have shown that $\LKp $and $ \LKpp$ simulate quantifier shifts in the cut-free format. The following theorem follows as a consequence of this fact

\begin{theorem}(cf. Aguilera, Baaz\cite{DBLP:journals/jsyml/AguileraB19}) \label{mainqs}
There is no elementary function bounding the length of the shortest cut-free $ \LK$-proof of a formula in terms of its shortest cut-free $\lkqs$-proof.
\end{theorem}
\begin{proof}
We will make use of the family of sequents $\{S_i\}_{i < \omega}$ described in Theorem \ref{th3} %and due to Statman \cite{statman1979lower}, and specific $ \LK$-proofs thereof. The sequents and the proofs themselves are not important for our proof. What is relevant is that they have the following properties:
%\begin{enumerate}
%\item the size of $S_i$ is polynomial in $i$;
%\item there is no bound on the size of their smallest cut-free $ \LK$-proofs that is elementary in $i$;
%\item the size of these proofs (with cuts), however, is polynomially bounded in $i$;
%\item all cut formulae are closed; we can also assume they are prenex by, e.g., \cite[Theorem 3.3]{DBaazLeitsch}.
%\item the total amount of quantifiers in cut formulae in each proof is linear in $i$.
%\item each free variable is the characteristic variable of at most one strong-quantifier inference.
%\end{enumerate}
%Let $\Gamma_i \vdash \Delta_i$ be one of the sequents.

%We transform the cuts into closed prenex form according to Baaz, Leitsch \cite{DBLP:journals/apal/BaazL99 }.

We are left with a cut-free proof $\pi_0$ whose end sequent is of the form:
\begin{align} \label{eqSpeedup1}
A_0 \to A_0, \hdots, A_m \to A_m, \Gamma_i \vdash \Delta_i.
\end{align}

Choose any occurrence of a quantifier in $A_0$ that is not in the scope of another quantifier. Since it appears once in the antecedent of $A_0 \to A_0$ and once in the consequent, it appears once with each polarity. Because $A_0$ is assumed to be prenex, we can apply two deep quantifier shifts, so that we are left with a proof $\pi_0'$ of the sequent:
\begin{align} \label{eqSpeedup1.5}
\forall x^0_0\, \exists x^0_1\, \big(A_0' \to A_0'\big), \hdots, A_m \to A_m, \Gamma_i \vdash \Delta_i,
\end{align}
Continue choosing any occurrence of a quantifier in $A_0'$ that is not in the scope of another quantifier (within $A_0'$) and applying deep quantifier shifts in pairs until we obtain a proof $\pi_1$ of the sequent:
\begin{align*} 
\forall x^0_0\, \exists x^0_1\, \hdots \big(\hat A_0 \to \hat A_0\big), \hdots, A_m \to A_m, \Gamma_i \vdash \Delta_i,
\end{align*}
where $\hat A_0$ is quantifier-free and the prefix of $\hat A_0 \to \hat A_0$ consists of alternating quantifiers. Repeat this procedure for each of the $A_i$ to obtain a proof $\pi_m$ where each formula $A_i \to A_i$ is replaced by an expression of the form:
\begin{align*}
\forall x^i_0\, \exists x^i_1\, \hdots \big(\hat A_i \to \hat A_i\big),
\end{align*}
where $\hat A_i$ is quantifier free and the quantifier prefix is alternating.
Let $\hat\Gamma^i \vdash \hat\Delta^i$ be the sequent:
\begin{align} \label{eqSpeedup2}
\forall x^0_0\, \exists x^0_1\, \hdots \big(\hat A_0 \to \hat A_0\big), \hdots, \forall x^m_0\, \exists x^m_1\, \hdots \big(\hat A_m \to \hat A_m\big), \Gamma_i \vdash \Delta_i.
\end{align}
Note that, since the size of the initial proof was polynomial in $i$, there were polynomially many quantifiers in the proof; hence, we added only polynomially many deep quantifier shifts, and so the size of the resulting proof $\pi_m$ of \eqref{eqSpeedup2} is bounded polynomially in $i$. Moreover, it is cut-free. Consequently, it suffices to show:

\begin{claim} \label{ClaimSU}
There is no elementary function bounding the size of the smallest cut-free $ \LK$-proofs of \eqref{eqSpeedup2}.
\end{claim} 
\begin{proof}
    
Let $\{\sigma_i\}_{i < \omega}$ be a sequence of such proofs. First, we transform it into a sequence of proofs of the Skolemizations of the sequents \eqref{eqSpeedup2}, so that for each sequent, each of the implications $\hat A_i \to \hat A_i$ remains only universally quantified. Then, by Herbrand's theorem, there are propositional proofs $\{\theta_i\}_{i < \omega}$ of the Herbrand sequents of \eqref{eqSpeedup2}, each of which is of the form:
\begin{align} \label{eqSpeedup3}
\bigwedge_{k}B^0_k, \hdots, \bigwedge_{k}B^m_k, \Gamma' \vdash \Delta'.
\end{align}
Moreover, the lengths of the Herbrand sequents are bounded exponentially by the lengths of $\{\theta_i\}_{i < \omega}$ (see \cite{DBaazLeitsch}). 
Each conjunct $B^p_k$ is a quantifier-free implication of the form:
\begin{align} \label{eqSpeedup4}
A(t_1, \hdots, t_l) \rightarrow A(s_1, \hdots, s_l),
\end{align}
for some terms $t_1, \hdots, t_l, s_1, \hdots, s_l$. The key point is that, in the process of Skolemizing the sequents \eqref{eqSpeedup2}, exactly one term in each pair $(t_j, s_j)$ was weakly quantified and the other (which was under the scope of the strong quantifier) was replaced by a Skolem function. Hence, either $t_j$ is of the form $f(s_j, \vec r)$ or $s_j$ is of the form $f(t_j, \vec r)$, for some tuple of terms $\vec r$.

Transform the sequent \eqref{eqSpeedup3} as follows: pick a pair $(t_j, s_j)$ of some disjunct $B^p_k$ such that the element thereof that was replaced by a Skolem term---say, $f(s_j, \vec r) = t_j$---does not have any other of $t_1, \hdots, t_l, s_1, \hdots, s_l$ as an argument. Such a term of course exists---it is the first term whose quantifier was  shifted in \eqref{eqSpeedup1.5}. Substitute $s_j$ for $t_j$ throughout the sequent. 

By repeating this process sufficiently-many times, each formula \eqref{eqSpeedup4} is transformed into a propositional tautology of the form $A(\vec t) \to A(\vec t)$. This means that the sequent
\begin{align*}
\vdash \bigwedge_{k}B^0_k, \hdots, \bigwedge_{k}B^m_k.
\end{align*}
is also a propositional tautology, whereby so too is
$\Gamma' \vdash \Delta'$,
from which follows that it is $ \LK$-provable. By \cite{Hetzl}, there is an $ \LK$-proof of the unskolemized sequent of length exponential in that of $\Gamma' \vdash \Delta'$. But this is impossible, since---by assumption---there exist no short proofs of the unskolemized sequent.
\end{proof}

%\noindent This establishes Theorem \ref{mainqs} and thus Theorem \ref{main}.
%\qed
\end{proof}

\begin{corollary}\label{cor.2}
There is a specific family of sequents $\{S_i \}_{i<\omega}$ and specific $\LKpp$-proofs such that they have the following properties:
    \begin{enumerate}
        \item the size of $S_i$ is elementary in $i$, and  $S_i$ contains only weak quantifiers;
        \item there is no bound on the size of their smallest cut-free $\LKpp$-proofs that is elementary in $i$;
        \item the size of these proofs (with cuts), however, is elementary bounded in $i$;
        \item all cut formulas are closed and prenex.
    \end{enumerate}\end{corollary}

\begin{proof}
    Note that cut-free proofs of sequents with weak quantifiers only coincide for $\lk,\LKp$ and $\LKpp$  as shown in Theorem \ref{th.3}.\qed
\end{proof}

\begin{corollary}
There is no elementary structural Skolemization of $\LKpp$-proofs.     
\end{corollary}

\begin{proof}
    As cut-free proofs with weak quantifiers only coincide for $ \LK$ and $\LKpp$ proofs, this would lead to a contradiction as $ \LK$ proofs are for sure $\LKpp$ proofs. \qed
\end{proof}

In consequence, we obtain as further results that there is non-elementary difference between resolution-refutations under structural and Andrews Skolemization (note that cut-free $ \LK$ and $\LKpp$ proofs can be elementarily transformed into ground resolution-refutation and vise versa \cite{BaazLolic}).

%In the sections before, we have seen that certain ways of Skolemization are associated with certain concepts of cut-free proofs in an elementary way. 

\section{More Efficient Proof Systems in the Cut-free Context}

By considering weaker and weaker side-variable conditions, we obtain calculi with shorter and shorter cut-free proofs (for example, a cut-free proof may simulate not only quantifier shifts but also quantifier distributions. In this case, the elementarily associated Skolemization can be found in Baaz, Lolic\cite{Effective}). Note that the sequent calculi always appear in pairs depending on weak and very weak regularity. \\

Epsilon calculus in its sequent formulation allows, however, for a non-elementary speed-up on cut-free $\LKp$ and $\LKpp$ proofs. \\

The language of the formulas in     $\LE$ is the language of $\varepsilon$-calculus.\\
The rules of     $\LE$ are the propositional rules of $ \LK$ expanded by the following quantifier rules:

\begin{itemize}
    \item the weak quantifier inferences $\exists_r$
\begin{prooftree}
	\AxiomC{$\Pi \vdash \Delta, A(t)$}
	\RightLabel{$\exists_r$}
	\UnaryInfC{$\Pi \vdash \Delta, A(\varepsilon_x A(x))$}
\end{prooftree}

	and $\forall_l$
\begin{prooftree}
	\AxiomC{$A(t), \Pi \vdash \Delta$}
	\RightLabel{$\forall_l$}
	\UnaryInfC{$A(\varepsilon_x \neg A(x)), \Pi \vdash \Delta$}
\end{prooftree}
 \item the strong quantifier inferences: \\[1ex]
$\exists_l$: substitution of $\varepsilon_x A'(x)$ for the characteristic variable\\[1ex]
$\forall_r$: substitution of $\varepsilon_x \neg A'(x)$ for the characteristic variable.

\end{itemize}

\begin{definition}
An epsilon term $\varepsilon_x A(x)$ overbinds if it is generated by the following inferences
$$A(f(e_1, \ldots, e_n), g(e_1, \ldots, e_n)) \to A(\varepsilon_x A(x, g(e_1, \ldots, e_n),g(e_1, \ldots e_n)),$$
\begin{prooftree}
	\AxiomC{$\Pi \vdash \Delta, A(t)$}
	\RightLabel{$\exists_r$}
	\UnaryInfC{$\Pi \vdash \Delta, A(\varepsilon_x A(x))$}
\end{prooftree}
\begin{prooftree}
	\AxiomC{$A(t), \Pi \vdash \Delta$}
	\RightLabel{$\forall_l$}
	\UnaryInfC{$A(\varepsilon_x \neg A(x)), \Pi \vdash \Delta$}
\end{prooftree}
\end{definition}

\begin{proposition} \label{prop.1}
Every $ \LK$-derivation possibly with cuts can be translated into an $\LE$-derivation of double exponential size.
\end{proposition}
\begin{proof}
All inference steps are replaced by corresponding inference steps with exception of strong quantifier rules, which are replaced by substitution.
This means, that every bound variable is replaced by an epsilon expression at most of the size of the existing derivation.\qed
\end{proof}
\begin{remark}
Note that the usual form of epsilon proofs can be obtained by deleting the quantifier inferences of $\LE$, and replacing them by 
	\begin{prooftree}
		\AxiomC{$(\psi')$}
		\noLine
		\UnaryInfC{$\Pi' \vdash \Delta', A'(t)$}
		\AxiomC{$A'(\varepsilon_x A(x)) \vdash A'(\varepsilon_x A(x))$}
		\RightLabel{$\exists_r$}
		\BinaryInfC{$A'(t) \to A'(\varepsilon_x A(x)), \Pi' \vdash \Delta'$}
	\end{prooftree}
and
	\begin{prooftree}
		\AxiomC{$(\psi')$}
		\noLine
		\UnaryInfC{$A'(t), \Pi' \vdash \Delta'$}
		\AxiomC{$A'(\varepsilon_x \neg A(x)) \vdash A'(\varepsilon_x \neg A(x))$}
		\RightLabel{$\forall_l$}
		\BinaryInfC{$A'(\varepsilon_x \neg A(x)) \to A'(t), \Pi' \vdash \Delta'$}
		\doubleLine
		\UnaryInfC{$\neg A'(t) \to \neg A'(\varepsilon_x \neg A(x)), \Pi' \vdash \Delta'$}
	\end{prooftree}

\end{remark}
\begin{example}
Consider the sequent 
$$\exists y (A(y) \to \forall x A(x)).$$
Its shortest cut-free $ \LK$-derivation is

\begin{prooftree}
	\AxiomC{$A(a) \vdash A(a)$}
	\RightLabel{$w_r$}
	\UnaryInfC{$A(a) \vdash A(a), \forall x A(x)$}
	\RightLabel{$\to_r$}
	\UnaryInfC{$\vdash A(a), A(a) \to \forall x A(x)$}
	\RightLabel{$\exists_r$}
	\UnaryInfC{$\vdash A(a), \exists y (A(y) \to \forall x A(x))$}
	\RightLabel{$\forall_r$}
	\UnaryInfC{$\vdash \forall x A(x), \exists y (A(y) \to \forall x A(x))$}
	\RightLabel{$w_l$}
	\UnaryInfC{$A(b) \vdash \forall x A(x), \exists y (A(y) \to \forall x A(x))$}
	\RightLabel{$\to_r$}
	\UnaryInfC{$\vdash A(b) \to \forall x A(x), \exists y (A(y) \to \forall x A(x))$}
	\RightLabel{$\exists_r$}
	\UnaryInfC{$\vdash \exists y (A(y) \to \forall x A(x)), \exists y (A(y) \to \forall x A(x))$}
	\RightLabel{$c_r$}
	\UnaryInfC{$\vdash \exists y (A(y) \to \forall x A(x))$}
\end{prooftree}
When we translate the sequent to epsilon calculus we obtain
$$\vdash A(\varepsilon_y (A(y) \to A(\varepsilon_x \neg A(x)))) \to A(\varepsilon_x \neg A(x)),$$
and the translation of the above proof to $\LE$ is 

\begin{prooftree}
	\AxiomC{$A(\varepsilon_x \neg A(x)) \vdash A(\varepsilon_x \neg A(x))$}
	\RightLabel{$w_r$}
	\UnaryInfC{$A(\varepsilon_x \neg A(x)) \vdash A(\varepsilon_x \neg A(x)), A(\varepsilon_x \neg A(x))$}
	\RightLabel{$\to_r$}
	\UnaryInfC{$\vdash A(\varepsilon_x \neg A(x)), A(\varepsilon_x \neg A(x)) \to A(\varepsilon_x \neg A(x))$}
	\RightLabel{$\exists_r$}
	\UnaryInfC{$\vdash A(\varepsilon_x \neg A(x)), A(\varepsilon_y (A(y) \to A(\varepsilon_x \neg A(x)))) \to A(\varepsilon_x \neg A(x))$}
	\RightLabel{$(*) \ + w_l$}
	\UnaryInfC{$A(b) \vdash A(\varepsilon_x \neg A(x)), A(\varepsilon_y (A(y) \to A(\varepsilon_x \neg A(x)))) \to A(\varepsilon_x \neg A(x))$}
	\RightLabel{$\to_r$}
	\UnaryInfC{$\vdash A(b) \to A(\varepsilon_x \neg A(x)), A(\varepsilon_y (A(y) \to A(\varepsilon_x \neg A(x)))) \to A(\varepsilon_x \neg A(x))$}
	\RightLabel{$\exists_r$}
	\UnaryInfC{$\vdash A(\varepsilon_y (A(y) \to A(\varepsilon_x \neg A(x)))) \to A(\varepsilon_x \neg A(x)), A(\varepsilon_y (A(y) \to A(\varepsilon_x \neg A(x)))) \to A(\varepsilon_x \neg A(x))$}
	\RightLabel{$c_r$}
	\UnaryInfC{$\vdash A(\varepsilon_y (A(y) \to A(\varepsilon_x \neg A(x)))) \to A(\varepsilon_x \neg A(x))$}
\end{prooftree}
Note that the formula in the end-sequent above is $[\exists y (A(y) \to \forall x A(x))]^\varepsilon$, and for $(*)$: $\forall_r$ has been replaced by the substitution of $\varepsilon_x \neg A(x)$ for $a$.\\ [1ex]
The shortest cut-free derivation of the sequent
$$\vdash A(\varepsilon_y (A(y) \to A(\varepsilon_x \neg A(x)))) \to A(\varepsilon_x \neg A(x))$$
in $\LE$ is however

\begin{prooftree}
	\AxiomC{$A(\varepsilon_x \neg A(x)) \vdash A(\varepsilon_x \neg A(x))$}
	\RightLabel{$\to_l$}
	\UnaryInfC{$\vdash A(\varepsilon_x \neg A(x)) \to A(\varepsilon_x \neg A(x))$}
	\RightLabel{$\exists_r$}
	\UnaryInfC{$\vdash A(\varepsilon_y (A(y) \to A(\varepsilon_x \neg A(x)))) \to A(\varepsilon_x \neg A(x))$}
\end{prooftree}
\end{example}
%

%%%%%%%%%%%%%%%%%%%%5
%%%%%%%%%%%%%%%%%%%%%%%%
\begin{theorem}
    Consider      $\LE$-sequence which are translation of $ \LK$-sequence with weak quantifiers only. Cut-free proofs can be obtained by the extended first epsilon theorem.    
\end{theorem}
\begin{proof}
    Retranslate      $\LE$-sequent into epsilon calculus and apply extended first epsilon theorem. This leads to a valid propositional sequent without quantifiers. For this sequent a cut-free proof can be bound exponentially. Weak quantifiers can be reintroduced and formulas can be contracted if necessary. Note that the bound of size is elementary in the bounds of the extended first epsilon theorem.    \qed    
    \end{proof}

%%%%%%%%%%%%%%%%%%%%%%%%%%%%%%%%5
%
\begin{theorem}
There is a sequence of cut-free $\LE$-proofs of $S_i$ such that
\begin{enumerate}
	\item the size of $S_i$ is elementary in $i$;
	\item the end-sequents $S_i$ are translations of first-order sequents $S'_i$ with weak quantifiers only;%, i.e. $S_i = [S'_i]^\varepsilon$;
	\item the size of these $\LE$-proofs is elementarily bounded in $i$;
	\item There is no elementary bound on the size of the smallest {$\LKp$}-proofs of $S'_i$.
\end{enumerate}
\end{theorem}
\begin{proof}
We choose a sequence of $ \LKp$-proofs from Corollary \ref{cor.2} which has non-elementarily growing cut-free proofs. We add $\forall \bar{x}(A(\bar{x})\rightarrow A(\bar{x}))$ to the antecedent of the sequents, where $A(\bar{x})$ are the matrices of the cuts (they are eliminable at double-exponential expenses in $\LK$ and cut-free proofs with weak quantifiers only coincide for $\LK$ and $\LKp$).
We translate the $ \LKp$-proofs with cut into epsilon calculus. This lengthens the proof at most double exponentially. 
In the $\LE$-proof we replace all cuts on $A$ by inferences of $A \to A$ on the left side.
Then we replace the epsilon terms in $A \to A$ by epsilon terms representing the universal quantifiers to obtain 
$$[\forall \overline{x} (A^M \to A^M)]^\varepsilon$$
on the left side. Finally, we contract with the existing translations of universal formulas in the end-sequent.\qed
\end{proof}
Note that the extended first epsilon theorem provides an upper bound for cut-free $\LE$-derivations:
The number of weak quantifier inferences determines the number of critical formulas.
From the Herbrand expansion according to the extended first epsilon theorem we obtain easily a cut-free $ \LK$ derivation.

%The question remains however, whether $\LE$-derivations with cuts can be translated into $ \LK$-derivations with arbitrary cuts in an elementary way.(posed by T. Arai)

\begin{example}
This example shows how $\LE$ reduces cuts of arbitrary complexity to universal cuts. Consider the following proof in $ \LK$

\begin{prooftree}
	\AxiomC{$P(a,b) \vdash P(a,b)$}
	\RightLabel{$\forall_l$}
	\UnaryInfC{$\forall y P(a,y) \vdash P(a,b)$}
	\RightLabel{$\forall_l$}
	\UnaryInfC{$\forall x \forall y P(x,y) \vdash P(a,b)$}
	\RightLabel{$\forall_r$}
	\UnaryInfC{$\forall x \forall y P(x,y) \vdash \forall y P(a,y)$}
	\RightLabel{$\exists_r$}
	\UnaryInfC{$\forall x \forall y P(x,y) \vdash \exists x \forall y P(x,y)$}
	\AxiomC{$P(c,d) \vdash P(c,d)$}
	\RightLabel{$\exists_r$}
	\UnaryInfC{$P(c,d) \vdash \exists y P(c,y)$}
	\RightLabel{$\exists_r$}
	\UnaryInfC{$P(c,d) \vdash \exists x \exists y P(x,y)$}
	\RightLabel{$\forall_l$}
	\UnaryInfC{$\forall y P(c,y) \vdash \exists x \exists y P(x,y)$}
	\RightLabel{$\exists_l$}
	\UnaryInfC{$\exists x \forall y P(x,y) \vdash \exists x \exists y P(x,y)$}
	\RightLabel{$cut$}
	\BinaryInfC{$\forall x \forall y P(x,y) \vdash \exists x \exists y P(x,y)$}
\end{prooftree}
The cut is replaced by an application of $\to_l$

\begin{prooftree}
	\AxiomC{$P(a,b) \vdash P(a,b)$}
	\RightLabel{$\forall_l$}
	\UnaryInfC{$\forall y P(a,y) \vdash P(a,b)$}
	\RightLabel{$\forall_l$}
	\UnaryInfC{$\forall x \forall y P(x,y) \vdash P(a,b)$}
	\RightLabel{$\forall_r$}
	\UnaryInfC{$\forall x \forall y P(x,y) \vdash \forall y P(a,y)$}
	\RightLabel{$\exists_r$}
	\UnaryInfC{$\forall x \forall y P(x,y) \vdash \exists x \forall y P(x,y)$}
	\AxiomC{$P(c,d) \vdash P(c,d)$}
	\RightLabel{$\exists_r$}
	\UnaryInfC{$P(c,d) \vdash \exists y P(c,y)$}
	\RightLabel{$\exists_r$}
	\UnaryInfC{$P(c,d) \vdash \exists x \exists y P(x,y)$}
	\RightLabel{$\forall_l$}
	\UnaryInfC{$\forall y P(c,y) \vdash \exists x \exists y P(x,y)$}
	\RightLabel{$\exists_l$}
	\UnaryInfC{$\exists x \forall y P(x,y) \vdash \exists x \exists y P(x,y)$}
	\RightLabel{$\to_l$}
	\BinaryInfC{$\forall x \forall y P(x,y), \exists x \forall y P(x,y) \to \exists x \forall y P(x,y)\vdash \exists x \exists y P(x,y)$}
\end{prooftree}
The translation of this derivation to $\LE$ is

\begin{prooftree}
	\AxiomC{$P(a,x_1(a)) \vdash P(a,x_1(a))$}
	\UnaryInfC{$P(x_2,x_1(x_2)) \vdash P(a,x_1(a))$}
	\UnaryInfC{$P(x_2,x_1(x_2)) \vdash P(x_3,x_1(x_3))$}
	\AxiomC{$P(x_3,d) \vdash P(x_3,d)$}
	\UnaryInfC{$P(x_3,d) \vdash P(x_3, x_4(x_3))$}
	\UnaryInfC{$P(x_3,d) \vdash P(x_5,x_4(x_5))$}
	\UnaryInfC{$P(x_3,x_1(x_3)) \vdash P(x_5,x_4(x_5))$}
	\BinaryInfC{$P(x_2,x_1(x_2)), P(x_3,x_1(x_3)) \to P(x_3,x_1(x_3)) \vdash P(x_5,x_4(x_5))$}
	\UnaryInfC{$P(x_2,x_1(x_2)), P(x_3,x_6(x_3)) \to P(x_3,x_6(x_3)) \vdash P(x_5,x_4(x_5))$}
	\UnaryInfC{$P(x_2,x_1(x_2)), P(x_7,x_6(x_7)) \to P(x_7,x_6(x_7)) \vdash P(x_5,x_4(x_5))$}
\end{prooftree}
Note that 
$$[\forall x \forall y (P(x,y) \to P(x,y))]^\varepsilon) = P(x_7,x_6(x_7)) \to P(x_7,x_6(x_7)),$$
and 
\begin{eqnarray*}
x_1(a) &=& \varepsilon_y \neg P(a,y),\\
x_2 &=& \varepsilon_x \neg P(x,x_1(x)),\\
x_3 &=& \varepsilon_x P(x,x_1(x)),\\
x_4(c) &=& \varepsilon_x P(c,x),\\
x_5 &=& \varepsilon_x P(x,x_4(x)),\\
x_6(d) &=& \varepsilon_y \neg P(d,y),\\
x_7 &=& \varepsilon_x \neg (P(x,x_6(x)) \to P(x,x_6(x)). \\
\end{eqnarray*}
\end{example}

\begin{remark}
    The proofs of the theorems above show that no cut-elimination with induction on the complexity of the cut formulas is feasible. In consequence, the first epsilon theorem is an argument of a completely different category.
\end{remark}

%%%%%%%%%%%%%%%%%%%%%%%%%%%%%
This section demonstrates that $\LKp$, $\LKpp$ incorporate only some properties of epsilon calculus, namely the occurrence behavior of epsilon terms, however, the full system is stronger then these calculi by the the possibility to overbind epsilon terms by other epsilon terms. This property is however, in contrast to usual unification so no correspondingly resolution-refutation system is known.

\section{Suitable Quantifier Inferences for Non-classical Logics}

Non-classical versions of $\LKp, \LKpp$ are not sound if the logics do not contain all  classical quantifier shift principles. In addition, they do not admit cut-elimination as the following example shows:

\begin{example}\label{unsound-ex}
Consider the following $\LJpp$-proof: %\footnote{Since only one formula appears on the right-hand side of each sequent. \rah{I would also delete this, it's obvious.}} 
% with cuts that includes existence property:
  \begin{align} \label{example2}
\infer[]{\vdash \exists x\, \big( A(x) \rightarrow A(f(x))\big)}{
\infer[]{\vdash A(a) \to A(f(a))}{
\infer[]{A(a) \vdash A(f(a))}{
	\infer[]{A(a) \vdash \forall x\, A(x)}{A(a)\vdash A(a)}
	&
	\infer[]{\forall x\, A(x) \vdash A(f(a))}{A(f(a)) \vdash A(f(a))}
}
}
}
\end{align}

%Since only one formula appears on the right-hand side of each sequent, this is an $\LJpp$-proof.
If there were a cut-free proof of
\begin{equation} \label{eqcutelimination}
\vdash \exists x\, \big( A(x) \rightarrow A(f(x))\big),
\end{equation} 
by the subformula property, this derivation would consist entirely of subformulas of \eqref{eqcutelimination} and would contain no strong quantifier inferences. Therefore, it would already be an $\LJ$-proof. This is impossible, as \eqref{eqcutelimination} is not intuitionistically valid. %\rah{Why isn't it intuitionistically valid?} \mar{$ \exists x\, \big( A(x) \rightarrow A(f(x))\big)$ is not valid in $\iqc$, therefore it is an example that $\LJpp$ derives unsound results in $\iqc$ with cuts}
\end{example}

\begin{remark}
This example demonstrates, in addition,  that Gentzen-style cut-elimination is not feasible for $\LKp$ and     $\LKpp$, as Gentzen-style cut-elimination transforms intuitionistic derivation automatiacally in intuitionistic derivations.
\end{remark}

In this section, we will concentrate on intuitionistic logic and intermediate extensions. In intuitionistic logic, the following three quantifier shifts are not valid:

%LJ+ and LJ++ the intuitionistic versions of LK+ and LK++ are not sound. The reason is that  [insert]

\begin{itemize}
    \item
 Constant Domain (CD):
    \[\forall x (A(x) \vee B) \to \forall x A(x) \vee B\]
    \item
   Independence of Premise (IP) 
   \[(B \to \exists x A(x)) \to \exists x (B \to A(x))\]
\item
    Quantifier Switch (SW):
    \[(\forall x A(x) \to B) \to \exists x (A(x) \to B)\]
   % If we consider the commuting of quantifiers $\forall, \exists$ with the connectives $\wedge, \vee, \to$, then the quantifier shift formulas are the only ones not valid in $\IQC$.
\end{itemize}

We denote by $\LJp$ and $ \LJpp$ intuitionistic version of $\LKp, \LKpp$.  $\LJp$ and $ \LJpp$ are not sound for intuitionistic logic.

$\LJp$ and $ \LJpp$ become sound if the lacking second and third quantifier shift instances are added as axioms. (The first shift follows from the other two see Baaz, Gamsakhurdia, Iemhoff, Jalali\cite{Skolemization-apal}). This does not change anything 
in the non-eliminability of cuts in any format.

Note that by the example above, no cut-free proof system with the subformula property is possible. The concept, however, can be used to characterize intermediate logics with Andrews Skolemization and deSkolemization.% (and in consequence structural Skolemization and deSkolemization). 

An intermediate logic is considered to be given by finitely many axiom schemes and the axiom schemes and rules of intuitionistic logic.

\begin{theorem} (See \cite{Skolemization-apal})
    An intermediate first-order logic admits structural and Andrews Skolemization if all quantifiers shifts are available.
\end{theorem}

\begin{proof}
    Note that deSkolemization is not unique. Therefore, all quantifier shifts can be generated by Skolemizing and deSkolemizing.\qed
\end{proof}

Note that the deSkolemization process is inherently non-deterministic, yet equivalent, reconstructed proofs may result depending on how Skolem functions are replaced and quantifiers are reintroduced. This flexibility does not affect correctness or completeness, as the logical structure and dependencies are preserved in all cases.

%Consider the shift rule $\CD= \forall x (A(x) \vee B) \to \forall x (A(x) \vee B)$
%in Skolemization is \[
%\CD^S = \forall x (A(x) \vee B) \to A(c) \vee B,
%\]
%and as deSkolemization, we may well choose 

%\[\forall x (A(x) \vee B) \to \forall x A(x) \vee B\]
%Similar for the other shift rules. Therefore, all quantifier shift rules follow from the soundness of Skolemization and deSkolemization. 

\begin{theorem}
Let $L$ be an intermediate logic which is sound for Andrews Skolemization and deSkolemization, then $L$ derives all classical quantifier shifts.
\end{theorem}

\begin{proof}
Skolemization is sound by substitution, therefore we only have to prove that
$L \vdash\CD $, $L \vdash\ED $, $L \vdash\SW $.

$\bullet$ The case $L \vdash\CD $. It is easy to see that the only strong quantifier in $\CD= \forall x (A(x) \vee B) \to \forall x (A(x) \vee B)$ is the universal quantifier in the succedent of the implication, which is not in the scope of any weak quantifiers. Therefore,
\[
\CD^S = \forall x (A(x) \vee B) \to A(c) \vee B,
\]
where $c$ is a fresh constant symbol not occurring in $\CD$. It is easy to see that $\iqc \vdash \CD^S$. Thus, $L \vdash \CD^S$, and by the assumption we have $L \vdash \CD$.

$\bullet$ For the other two cases, with similar reasoning, we have 
\begin{center}
$\ED^S= (B \to A(c)) \to \exists x (B \to A(x))$, \\ 
$\SW^S= (A(c) \to B) \to \exists x (A(x) \to B)$,
\end{center}
where $c$ is a fresh constant not appearing in $A(x)$ and $B$. As $\ED^S$ and $\SW^S$ are valid in $\iqc$, thus $L \vdash \ED$ and $L \vdash \SW$.
%Let $\mathsf{IQC} \subseteq L \subseteq \mathsf{CQC}$. If $L \vdash A^S$ then $L \vdash A$ holds, i.e.,  $L$ enjoys \emph{deSkolemization}, then $L \vdash \CD$, $L \vdash \SW$, and $L \vdash \ED$. We prove the lemma for each case. For $\CD$, we have \[\CD^S= \forall x (A(x) \vee B) \to A(c) \vee B,\]where $c$ is a fresh constant not appearing in $A(x)$ and $B$. Since $L \vdash \forall x \phi(x) \to \phi(c)$, for any formula $\phi(x)$, we have $L \vdash \CD^S$, when $\phi(x)$ is substituted by $A(x) \vee B$. By the hypothesis, we have $L \vdash \CD$.

%For $\SW$, we have \[\SW^S = (A(c)\to B)\to \exists x (A(x) \to B),\]where $c$ is a fresh constant not appearing in $A(x)$ and $B$. As $L \vdash \phi(c) \to \exists x \phi(x)$, for any formula $\phi(x)$, we have $L \vdash \SW^S$, when $\phi(x)$ is substituted by $A(x) \to B$. By the hypothesis, we have $L \vdash \SW$.

%For $\ED$, we have \[\ED^S = (B \to A(c)) \to \exists x (B \to A(x)),\]where $c$ is a fresh constant not appearing in $A(x)$ and $B$. As $L \vdash \phi(c) \to \exists x \phi(x)$, for any formula $\phi(x)$, we have $L \vdash \ED^S$, when $\phi(x)$ is substituted by $A(x) \to B$. By the hypothesis, we have $L \vdash \ED$.
Note that the deSkolemization of quantifier shift principles is the same for structural and Andrews Skolemization.\qed
\end{proof}
We now intend to show the soundness of the deSkolemization in the presence of all quantifies shift rules, i.e., lacking ones.

\begin{theorem}
    Let $L$ be an intermediate logic deriving all classical quantifier shift rules then $L$ is sound for Andrews Skolemization and deSkolemization.
\end{theorem}

\begin{proof}
 Skolemization is sound by substitution, deSkolemization follows from the following:
Assume in $L$ a Skolemized form $A^s(a)$ is provable using axioms $A_1\dots A_n$. Then  $A_1\dots A_n\vdash A^s$ is derivable in $\lj$. Then $A^s_1\dots A^s_n\vdash A^s$ is derivable in $\lj$ by substitution again. Then we eliminate cuts in this $\lj$-proof. This cut-free $\lj$-proof is also a cut-free $\LJpp$-proof. Now we introduce strong quantifiers exactly as the deSkolemization in classical logic (cf. end of Section \ref{speed}). The result is a cut-free $\LJpp$-proof of $A_1\dots A_n\vdash A$. The resulting proof is unsound for intuitionistic logic, however, it can be made sound by adding quantifier shifts principles to the antecedent of the end-sequent. This is analogous to the soundness proof of $\LKpp$(see Theorem \ref{TheoremCorrectness}) using quantifier shifts which are derivable in $L$.\qed
\end{proof}

\begin{remark}
    Note that $\LJpp$ does not provide a cut-free system for intuitionistic logic expanded with quantifier shifts.
\end{remark}

Therefore, there is an extension of the second epsilon theorem for intermediate logic which derive all classical shifts. The price to pay is however that the construction in this theorem does not lead to cut-free proofs. Furthermore, the $\LJpp$-proofs lead to expansions as usual, which can be considered as extensions of the Herbrand Theorem in the sense that their soundness is not based on propositional validity only but is derived from variants of the shift principles. This is not uncommon in logics, cf. the number theoretic Herbrand disjunction in the analysis of the theorem of Roth which is validitated by Peano arithmetic (Luckhardt \cite{Luckhardt1989-LUCHZB}).

%%%%%%%%%%%%%%%%%%%%%%%%%%%%%%%%%%%%%
%%%%%%%%%%%%%%%%%%%%%5

%For the other direction, we argue as follows: the derivation of any Skolemized formula $A$ from formulas $P$ in the intermediate logic, is represented by the derivation of $P \limp A$ in $\LJ$. We Skolemize proof with end-sequent and apply cut-elimination. We deSkolemize the cut-free proof w.r.t Andrews Skolemization. A cut-free proof with structura Skolemization can be easily transformed to cut-free proof with Andrews Skolemization by deleting positions of the Skolem terms. Any $\LJ$ proof is an $ \LJpp$ proof and we reintroduce quantifiers as in (Section about Deskolemization). This proof might be incorrect but it becomes correct in the same way as $\LKp$ and $\LKpp$ are shown to be \textbf{sound \cite{}}. If it is possible to cut with consequences of the quantifier shift axioms. Note that the constructed derivation is not analytic. This is an example for a application of a proof theory in a non-analytic context. 

\begin{remark}
    The same argumentation holds for most logics with an underlying analytic calculus whose + and ++ variants are given. Consequently, an extension of $S4$ admits Skolemization and deSkolemization iff all classical quantifier shifts are present. That is iff $\forall$ and $\exists$ can be shifted through $\Box$. 
\end{remark}

%\subsection{Conservativity of Epsilon Calculus in Non classical Logics}

As a consequence of the soundness of Skolemization and deSkolemization, we obtain, in addition, the conservativity of epsilon calculus understood as addition of critical formulas to the quantifier-free parts. Here, we deal with non-classical logics and, therefore, we have to represent 
 $\forall x A(x)$  by  $A(\tau _x(A))$ and critical fomrulas $A(\tau_xA(x))\rightarrow A(t)$  (cf. Bell \cite{Bell2006}).

First, we analyze the structure of epsilon terms.

\begin{definition}
An $\varepsilon$-term (no variables are bound from outside the term) $e$ is included in an $\varepsilon$-term $e'$ if $e$ is a proper subterm of $e'$. 

An $\varepsilon$-term $e$ is subordinate to an $\varepsilon$-term $e' = \varepsilon_x A(x)$ if $e$ occurs in $e'$ and $x$ occurs in $A$.

 A ground function of an epsilon term is the skeleton remaining when all included subterms are removed. 

The same definition holds for $\tau$.

\end{definition}

\begin{example}
    Consider the $\varepsilon$-term:
\[
e:=\varepsilon_ x\, A(x, \varepsilon _y\, B(x,y), \varepsilon _z\, C(z,0) )
\]

for some constant $0$. \\

$ \varepsilon _y\, B(x,y)$ is subordinated to $e$.\\

$\varepsilon _z\, C(z,0)$ is are included in $e$.\\
$0$ is included, but $\varepsilon _z\, C(z,0)$ is the only maximal term included in $e$.\\
The ground function is \[
\varepsilon_ x\, A(x, \varepsilon _y\, B(x,y), u )
\]

\end{example}

Ground functions can be identified with Skolem functions, composed epsilon terms can be identified with Skolem terms, therefore epsilon terms can be translated into Skolem terms.

\begin{example}
The corresponding Skolem term for the above example is $f(g(0))$.
\end{example}

We now sketch the argument of the conservativity of epsilon calculus.

\begin{theorem}
    For every intermediate logic, where structural  Skolemization and deSkolemization is sound, translation to epsilon calculus adding critical formulas is conservative.
\end{theorem}

\begin{proof}(Sketch.)

$\bullet$ Consider a proof of a formula in epsilon calculus which is an epsilon-translation of a formula with weak quantifiers only. Let $\IE$ be defined as $\LE$ with all sequents restricted at most to one on right side and $\forall_l$
\begin{prooftree}
	\AxiomC{$A(t), \Pi \vdash \Delta$}
	\RightLabel{$\forall_l$}
	\UnaryInfC{$A(\varepsilon_x \neg A(x)), \Pi \vdash \Delta$}
\end{prooftree}
 
replaced by

\begin{prooftree}
	\AxiomC{$A(t), \Pi \vdash \Delta$}
	\RightLabel{$\tau_l$}
	\UnaryInfC{$A(\tau _x A(x)), \Pi \vdash \Delta$}
\end{prooftree}

We obtain a derivation of a sequent with al critical formulas in the antecedent and the epsilon-translation of the formula in the succedent. Remove all propositional cuts. Now we replace all epsilon terms by Skolem terms using the identification of ground functions on Skolem functions depending on the maximal included subterms. We are now in $\LJ$.

$\bullet$ Introduce weak quantifiers to reconstruct original end-sequent.

$\bullet$ 
Introduce existential quantifier on the left (right side fo the transformed critical formula) if the original inference has be on the right (left) side.

$\bullet$  Replacing the term $t$ by its bound variables, where $t$ is a term replaced by epsilon term in the critical formula.

$\bullet$
Introduce $\forall$ on transformed critical formula for all maximal included terms. 

Confer the example: 
\[A(t,0) \rightarrow A(\varepsilon_x A(x,s),s),\, \Pi\vdash \Gamma\]
\[\Downarrow\]
\[A(t',0) \rightarrow A(f(s'),s'),\, \Pi'\vdash \Gamma'\]
\[\Downarrow\]
\[\exists y A(y,0)\rightarrow A(f(s'),s')\]
\[\Downarrow\]
\[\forall x(\exists y A(y,x)\rightarrow A(f(x),x))\]

$\bullet$ There is a Skolem function not ocurring in the domain of other Skolem functions (this is guaranteed by the strictness of the order of subordination). DeSkolemize with respect to these Skolem functions.

We obtain according to the example
\[\forall x(\exists y A(x,y)\rightarrow \exists y A(y,x)), \Pi'\vdash \Gamma'\]

$\bullet$
Derive the transformed critical formula in $\LJ$ and cut it with the corresponding formula in the antecedent. 

$\bullet$
Eliminate these cuts. (This is necessary as a transformed critical formula might contain other Skolem functions).

$\bullet$
Repeat the procedure until exhausting the Skolem functions.

%Consider the translation into epsilon calculus where all epsilon terms are replaced by Skolem terms in the sense above. Now we have to show that this is in fact conservative.

%Now choose a translated critical formula\[A(t) \limp A(s)\] where $s$ is a Skolem term and where no Skolem function occurs in any matrix $ P(x)$ of any critical formula. Such a critical formula exists because of the subordination of the epsilon terms which excludes circular behaviour. We infer 
\qed
\end{proof}

Note that more and more quantifeir shifts are added to the antecedents during the deSkolemization process.

\begin{remark}
    The conservativity of epsilon calculus in intermediate logics with quantifier shifts does not mean that the (extended) first epsilon theorem holds, which is feasible only for the finite-valued first-order G\"odel logics. 
\end{remark}

    %generality of the method for S4 ... does not matter lawats the epsilon calculus is sound

\section{Conclusion}

Another application of $\LKpp$ is the possibility to handle macros of quantifiers as a single quantifier of its own.
This is not possible in $\LK$ as shown in Baaz, Lolic \cite{DBLP:conf/wollic/BaazL19}. 
The same principle holds for higher-order systems and corresponding partial systems, e.g., for Henkin quantifiers defined in the sense of

\[
\left(
\begin{array}{ll}
\forall x & \exists u \\
\forall y & \exists v
\end{array}
\right)
\, A(x,y,u,v).
\]

being defined as second-order formula
$\exists f \,\exists g \,\forall x \,\forall y \; A(x,y,f(x),g(y))$, 
where $f$ and $g$ are function variables.

This demonstrates the adaption of $ \LK$ to achieve specific mathematical insights.
Maybe in future mathematical theories will be associated with their specific proof systems in proof theory in the same way as mathematical theories being associated with specific classes of models in model theory.

\bibliographystyle{abbrv}
\bibliography{References}

\end{document}